%% file: entropy.tex
\documentclass[letterpaper,11pt]{article}

\usepackage[margin=1in]{geometry}  % set the margins to 1in on all sides
\usepackage{xspace,exscale,relsize}
\usepackage{fancybox,shadow}
\usepackage{graphicx}
\usepackage{color}
\usepackage{amsthm,amsfonts}
\usepackage{amssymb}
\usepackage{authblk}
\usepackage[title]{appendix}
\usepackage[bottom]{footmisc}

\usepackage[noadjust]{cite}

\usepackage{amsmath}
    \makeatletter
    \let\over=\@@over \let\overwithdelims=\@@overwithdelims
    \let\atop=\@@atop \let\atopwithdelims=\@@atopwithdelims
    \let\above=\@@above \let\abovewithdelims=\@@abovewithdelims
    \makeatother
\usepackage{fixltx2e}
\usepackage{rotating}

\usepackage{ifpdf}

\usepackage{subfigure}
\usepackage{psfrag}

\usepackage{prettyref,enumerate}

\usepackage{tikz}
\usetikzlibrary{arrows}
\tikzstyle{int}=[draw, fill=blue!20, minimum size=2em]
\tikzstyle{dot}=[circle, draw, fill=blue!20, minimum size=2em]
\tikzstyle{init} = [pin edge={to-,thin,black}]

\usepackage[%dvips,
            CJKbookmarks=true,
            bookmarksnumbered=true,
            bookmarksopen=true,
            colorlinks=true,
            citecolor=red,
            linkcolor=blue,
            anchorcolor=red,
            urlcolor=blue
            ]{hyperref}

\newcommand{\floor}[1]{{\left\lfloor {#1} \right \rfloor}}
\newcommand{\Floor}[1]{\lfloor {#1} \rfloor}

\newcommand{\reals}{\mathbb{R}}
\newcommand{\naturals}{\mathbb{N}}
\newcommand{\integers}{\mathbb{Z}}
\newcommand{\Expect}{\mathbb{E}}
\newcommand{\expect}[1]{\mathbb{E}\left[#1\right]}
\newcommand{\Prob}{\mathbb{P}}
\newcommand{\prob}[1]{\mathbb{P}\left[#1\right]}

\newcommand{\TV}{{\sf TV}}

\newcommand{\expects}[2]{\mathbb{E}_{#2}\left[ #1 \right]}
\newcommand{\diff}{{\rm d}}
\newcommand{\eg}{e.g.\xspace}
\newcommand{\ie}{i.e.\xspace}
\newcommand{\iid}{i.i.d.\xspace~}

\newcommand{\linf}[1]{\left\|{#1} \right\|_{\infty}}

\newcommand{\Poi}{\mathrm{Poi}}
\newcommand{\poi}{\mathrm{poi}}
\newcommand{\Multinom}{\mathrm{Multinomial}}
\newcommand{\Binom}{\mathrm{Binomial}}
\newcommand{\abs}[1]{\left| #1 \right|} % Added by Pengkun
\newcommand{\pth}[1]{\left( #1 \right)}
\newcommand{\qth}[1]{\left[ #1 \right]}
\newcommand{\sth}[1]{\left\{ #1 \right\}}

\newcommand{\iiddistr}{{\stackrel{\text{\iid}}{\sim}}}
\newcommand{\inddistr}{\overset{\text{ind}}{\sim}} % added by Pengkun
\newcommand{\var}{\mathsf{var}}

\newcommand{\Bern}{\text{Bern}}

\newcommand{\indc}[1]{{\mathbf{1}_{\left\{{#1}\right\}}}}
\newcommand{\Indc}{\mathbf{1}}

\definecolor{myblue}{rgb}{.8, .8, 1}
\definecolor{mathblue}{rgb}{0.2472, 0.24, 0.6} % mathematica's Color[1, 1--3]
\definecolor{mathred}{rgb}{0.6, 0.24, 0.442893}
\definecolor{mathyellow}{rgb}{0.6, 0.547014, 0.24}

\newcommand{\tR}{{\tilde{R}}}

\newcommand{\sfP}{{\mathsf{P}}}

\newcommand{\calE}{{\mathcal{E}}}
\newcommand{\calF}{{\mathcal{F}}}

\newcommand{\calM}{{\mathcal{M}}}

\newcommand{\calP}{{\mathcal{P}}}

\newcommand{\calS}{{\mathcal{S}}}

\newcommand{\Th}{{^{\rm th}}}
\newcommand{\diverge}{\to \infty}

\newrefformat{eq}{(\ref{#1})}
\newrefformat{thm}{Theorem~\ref{#1}}
\newrefformat{th}{Theorem~\ref{#1}}
\newrefformat{chap}{Chapter~\ref{#1}}
\newrefformat{sec}{Section~\ref{#1}}
\newrefformat{algo}{Algorithm~\ref{#1}}
\newrefformat{fig}{Fig.~\ref{#1}}
\newrefformat{tab}{Table~\ref{#1}}
\newrefformat{rmk}{Remark~\ref{#1}}
\newrefformat{clm}{Claim~\ref{#1}}
\newrefformat{def}{Definition~\ref{#1}}
\newrefformat{cor}{Corollary~\ref{#1}}
\newrefformat{lmm}{Lemma~\ref{#1}}
\newrefformat{prop}{Proposition~\ref{#1}}
\newrefformat{pr}{Proposition~\ref{#1}}
\newrefformat{app}{Appendix~\ref{#1}}
\newrefformat{apx}{Appendix~\ref{#1}}
\newrefformat{ex}{Example~\ref{#1}}
\newrefformat{exer}{Exercise~\ref{#1}}
\newrefformat{soln}{Solution~\ref{#1}}
\def\unifto{\mathop{{\mskip 3mu plus 2mu minus 1mu%
    \setbox0=\hbox{$\mathchar"3221$}%
    \raise.6ex\copy0\kern-\wd0%
    \lower0.5ex\hbox{$\mathchar"3221$}}\mskip 3mu plus 2mu minus 1mu}}

\ifx\lesssim\undefined
\def\simleq{{{\mskip 3mu plus 2mu minus 1mu%
    \setbox0=\hbox{$\mathchar"013C$}%
    \raise.2ex\copy0\kern-\wd0%
    \lower0.9ex\hbox{$\mathchar"0218$}}\mskip 3mu plus 2mu minus 1mu}}
\else
\def\simleq{\lesssim}
\fi

\ifx\gtrsim\undefined
\def\simgeq{{{\mskip 3mu plus 2mu minus 1mu%
    \setbox0=\hbox{$\mathchar"013E$}%
    \raise.2ex\copy0\kern-\wd0%
    \lower0.9ex\hbox{$\mathchar"0218$}}\mskip 3mu plus 2mu minus 1mu}}
\else
\def\simgeq{\gtrsim}
\fi

\newtheorem{theorem}{Theorem}
\newtheorem{lemma}{Lemma}

\newtheorem{prop}{Proposition}

\theoremstyle{definition}

\newtheorem{remark}{Remark}

\newcommand{\Hplug}{\hat{H}_{\text{plug-in}}}
\newcommand{\Rplug}{R_{\text{\rm plug-in}}}

\newif\ifmapx
{\catcode`/=0 \catcode`\\=12/gdef/mkillslash\#1{#1}}
\edef\jobnametmp{\expandafter\string\csname contraction_apx\endcsname}
\edef\jobnameapx{\expandafter\mkillslash\jobnametmp}
\edef\jobnameexpand{\jobname}
\ifx\jobnameexpand\jobnameapx
\mapxtrue
\else
\mapxfalse
\fi

\begin{document}
\ifpdf
\DeclareGraphicsExtensions{.pgf}
\graphicspath{{figures/}}
\fi

% paper title
\title{Minimax rates of entropy estimation on large alphabets via best polynomial approximation}

\author{Yihong Wu and Pengkun Yang\thanks{The authors are with
the Department of Electrical and Computer Engineering and the Coordinated Science Laboratory, University of Illinois at Urbana-Champaign, Urbana, IL, \texttt{\{yihongwu,pyang14\}@illinois.edu}. 
 This paper was presented in part at the IEEE International Symposium on Information Theory, Hong Kong,
June, 2015 \cite{WY14-isit}.
This research was supported in part by the National Science Foundation under
Grant IIS-1447879, CCF-1423088, 
CCF-1527105, and
Strategic Research Initiative on Big-Data Analytics of the College of Engineering
at the University of Illinois.}}

\date{\today}

\maketitle

\begin{abstract}
Consider the problem of estimating the Shannon entropy of a distribution over $k$ elements from $n$ independent samples. We show that the minimax mean-square error is within universal multiplicative constant factors of 
$$\Big(\frac{k }{n \log k}\Big)^2 + \frac{\log^2 k}{n}$$
if $n$ exceeds a constant factor of $\frac{k}{\log k}$; otherwise there exists no consistent estimator.
This refines the recent result of Valiant-Valiant \cite{VV11} that the minimal sample size for consistent entropy estimation scales according to $\Theta(\frac{k}{\log k})$. The apparatus of best polynomial approximation plays a key role in both the construction of optimal estimators and, via a duality argument, the minimax lower bound.
\end{abstract}

%\tableofcontents

%\newpage
\section{Introduction}
\label{sec:intro}

%Denote $ [k]=\sth{1,\dots,k} $, suppose we have an unknown input distribution
%$ P=\pth{ p_1,\dots,p_k } \in \reals_+^k ~\text{s.t.}~ \sum_ip_i=1 $ on $ [k] $,
Let $P$ be a distribution over an alphabet of cardinality $k$.
Let $ X_1,\dots,X_n $ be \iid~samples drawn from $P $. 
Without loss of generality, we shall assume that the alphabet is $[k]\triangleq \sth{1,\dots,k}$.
To perform statistical inference on the unknown distribution $P$ or any functional thereof, a sufficient statistic is the histogram $ N\triangleq (N_1,\ldots,N_k)$, where
$$ N_j=\sum_{i=1}^n \indc{X_i=j}$$
records the number of occurrences of $j \in [k]$ in the sample. Then $ N\sim \Multinom(n,P) $.

The problem of focus is to estimate the Shannon entropy of the distribution $ P $:
%, where the functional $ H $ is defined as   
\[
H(P)=\sum_{i=1}^{k}p_i\log \frac{1}{p_i}.
\]
%The goal to find the lower bound on optimal sample size: the sample size that no consistent estimator exists, which means minimax risk (MSE)
To investigate the decision-theoretic fundamental limit, we consider the minimax quadratic risk of entropy estimation:
\begin{equation}
R^*(k,n)\triangleq \inf_{\hat{H}}\sup_{P \in \calM_k}\Expect_P[( \hat{H}(N)-H(P) )^2]     
    \label{eq:Rkn}
\end{equation}
where $\calM_k$ denotes the set of probability distributions on $[k]$.
The goal of the paper is
a) to provide a constant-factor approximation of the minimax risk $R^*(k,n)$,
b) to devise a linear-time estimator that provably attains $ R^*(k,n) $ within universal constant factors.

%Although various estimators for information measures has been proposed with provable guarantees, such as entropy \cite{Paninski03,CKV04}, Kullback-Leibler divergence and mutual information \cite{Paninski04,WKV09b}, directed information \cite{jiao2012universal}, little is known how close to the statistical optimality these procedures are in the minimax sense. To this end we need \emph{converse} results, \ie, \emph{minimax lower bounds}, which bound the minimax risk from below as a function of the sample size and the statistical model. Consider the minimax mean-square error of estimating directed information based on $n$ samples of the processes $X$ and $Y$ whose joint distribution belongs to the set $\calP$:
%\begin{equation}
%   R_n = \inf_{\hat{I}} \sup_{P_{XY} In \calP} \Expect[(\hat{I} - I(X \to Y))^2].
%   \label{eq:Rn}
%\end{equation}

Entropy estimation has found numerous applications across various fields, such as neuroscience \cite{spikes-book}, physics \cite{Vinck12}, telecommunication \cite{PW96}, biomedical research \cite{porta2001entropy}, etc.
Furthermore, it serves as the building block for estimating other information measures expressible in terms of entropy, such as mutual information and directed information, which are instrumental in machine learning applications such as learning graphical models \cite{CL68,QKC13,jiao2013di,Bresler15}.
%\cite{JVHW15,Paninski03}.

From a statistical standpoint, the problem of entropy estimation falls under the category of \emph{functional estimation}, where we are not interested in directly estimating the high-dimensional parameter (the distribution $P$) per se, but rather a function thereof (the entropy $H(P)$). 
%Lower bounds of the minimax risk $R_n$ provide us a benchmark for evaluating various causality learning algorithms as well as how many samples are necessary to reach a prescribed accuracy of estimating directed information.
Estimating a scalar functional has been intensively studied in nonparametric statistics, \eg, estimate a scalar function of a regression function such as linear functional \cite{Stone80,DL91}, quadratic functional \cite{CL05}, $L_q$ norm \cite{LNS99}, etc.
%or density estimation, 
% estimating linear functional of a regression function or a density (\eg, value at a point \cite{Stone80}, \cite{DL91}), quadratic functional \cite{CL05}, $L_q$ norm \cite{LNS99}. 
%PCA: lower dimensional functional still high-dimensional object.
To estimate a function, perhaps the most natural idea is the ``plug-in'' approach, namely, first estimate the parameter and then substitute into the function. 
This leads to the commonly used plug-in estimator, \ie, the empirical entropy, 
\begin{equation}
    \Hplug = H(\hat{P}),
    \label{eq:Hplug}
\end{equation}
where $\hat P=(\hat p_1,\ldots,\hat p_k)$ denotes the empirical distribution with $\hat p_i = \frac{N_i}{n}$.
As frequently observed in functional estimation problems, the plug-in estimator can suffer from severe bias (see \cite{Efron82,Berkson80} and the references therein). Indeed, although $\Hplug$ is asymptotically efficient and minimax (cf., \eg, \cite[Sections 8.7 and 8.9]{VdV00}),
%rate-minimax\footnote{The asymptotically minimaxity of $\Hplug$ in the large-sample limit was claimed in \cite{Paninski03} without providing a proof. In fact in this regime the rate-minimaxity $\Hplug$, that is, modulo constant factors, can be deduced from \prettyref{thm:main} and \prettyref{prop:Rplug-rate} below.} 
% (which, e.g., can be deduced from \prettyref{thm:main} and \prettyref{prop:Rplug-rate} below) 
 in the ``fixed-$k$-large-$n$'' regime, it can be highly suboptimal in high dimensions, where, due to the large alphabet size and resource constraints, we are constantly contending with the difficulty of \emph{undersampling} in applications such as
\begin{itemize}
        \item corpus linguistics: about half of the words in Shakespearean canon only appeared once \cite{ET76};
        \item network traffic analysis: many customers or website users are only seen a small number of times \cite{benevenuto2009characterizing};
        \item analyzing neural spike trains: natural stimuli generate neural responses of high timing precision resulting in a massive space of meaningful responses \cite{Berry13051997,mainen1995reliability,SLSKB97}.
\end{itemize}
%corpus linguistics (about half of the words in Shakespearean canon only appeared once \cite{ET76}), network traffic analysis (many customers or website users are only seen a small number of times \cite{benevenuto2009characterizing}), analyzing neural spike trains (natural stimuli generate neural responses of high timing precision resulting in a massive space of meaningful responses \cite{Berry13051997,mainen1995reliability,SLSKB97}). 

Statistical inference on large alphabets with insufficient samples 
has a rich history in information theory, statistics and computer science, with early contributions dating back to Fisher \cite{FCW43}, Good and Turing \cite{Good53},  Efron and Thisted \cite{ET76} 
and recent renewed interests on compression, prediction, classification and estimation aspects for large-alphabet sources \cite{OSZ04,BS09,KWTV13,WVK11,VV13}.
However, none of the current results allow a general understanding of the fundamental limits of estimating information quantities of distributions on large alphabets. 
The particularly interesting case is when the sample size scales \emph{sublinearly} with the alphabet size.

%In many situations estimating a lower-dimensional functional requires fewer samples, or equivalently, has a smaller minimax risk (or faster minimax convergence rate). This intuition, however, does not always hold. For example, estimating a covariance matrix (without any low-dimensional structures such as  sparsity or low-rankness) under spectral norm has the same rate as estimating its spectral norm \cite{CMW13}.
%
%Major observation:
%\begin{enumerate}[a)]
%    \item Bias dominates
%    \item Plug-in procedures can be highly suboptimal \cite{} for either test or estimation.
%
%\end{enumerate}

Our main result is the characterization of the minimax risk within universal constant factors:
\begin{theorem}
    \label{thm:main}
If $n \gtrsim \frac{k}{\log k}$,\footnote{For any sequences $\{a_n\}$ and $\{b_n\}$ of positive numbers, we write $a_n\gtrsim b_n$ or $b_n\lesssim a_n$ when $a_n\geq cb_n$ for some absolute constant $c$. Finally, we write $a_n \asymp b_n$ when both $a_n\gtrsim b_n$ and $a_n\lesssim b_n$ hold.} 
    then
    \begin{equation}
    R^*(k,n) \asymp \pth{\frac{k }{n \log k}}^2 + \frac{\log^2 k}{n} .
    \label{eq:main}
\end{equation}
If $n \lesssim \frac{k}{\log k}$, there exists no consistent estimators, \ie, $R^*(k,n) \gtrsim 1$.
\end{theorem}
To interpret the minimax rate \prettyref{eq:main}, we note that the second term corresponds to the classical ``parametric'' term inversely proportional to $\frac{1}{n}$, which is governed by the variance and the central limit theorem (CLT). The first term corresponds to the squared bias, which is the main culprit in the regime of insufficient samples. 
Note that $R^*(k,n) \asymp (\frac{k }{n \log k})^2$ if and only if $n \lesssim \frac{k^2}{\log^4 k}$, where the bias dominates. As a consequence, the minimax rate in \prettyref{thm:main} implies that to estimate the entropy within $\epsilon$ bits with probability, say 0.9, the minimal sample size is given by
\begin{equation}
n \asymp \frac{\log^2 k}{\epsilon^2} \vee \frac{k}{\epsilon \log k}.    
        \label{eq:sample-complexity}
\end{equation}

%Next we discuss the implication of \prettyref{thm:main} in the regime  of $\frac{k}{\log k} \lesssim n \lesssim k^C$ for any fixed $C>1$.

Next we evaluate the performance of plug-in estimator in terms of its worst-case mean-square error
\begin{equation}
    \Rplug(k,n) \triangleq \sup_{P \in \calM_k} \Expect_P[(\Hplug(N) - H(P))^2].
    \label{eq:Rplug}
\end{equation}
Analogous to \prettyref{thm:main} which applies to the optimal estimator, the risk of the plug-in estimator admits a similar characterization (see \prettyref{app:plug} for the proof):
% in the following proposition.
\begin{prop}
    \label{prop:Rplug-rate}
    If $n \gtrsim k$, then
    \begin{equation}
        \Rplug(k,n) \asymp \pth{\frac{k }{n}}^2 + \frac{\log^2 k}{n} .
        \label{eq:Rplug-rate}
    \end{equation}
    If $n \lesssim k$, then $\Hplug$ is inconsistent, \ie, $\Rplug(k,n) \gtrsim 1$.
\end{prop}
%We give a simple proof of \prettyref{prop:Rplug-rate} in \prettyref{app:plug}.
% using Pinsker's inequality and by tightening the previous analysis in \cite{AK01,Paninski03}.    
Note that the first and second term in the risk \prettyref{eq:Rplug-rate} again corresponds to the squared bias and variance respectively. 
While it is known that the bias can be as large as $\frac{k}{n}$ \cite{Paninski03}, the variance of the plug-in estimator is at most a constant factor of $\frac{\log^2 n}{n}$, regardless of the alphabet size (see, \eg, \cite[Remark (iv), p. 168]{AK01}). This variance bound can in fact be improved to $\frac{\log^2(k\wedge n)}{n}$ by a more careful application of Steele's inequality \cite{Jiao-email}, and hence the mean-square error (MSE) is upper bounded by $\pth{\frac{k}{n}}^2+\frac{\log^2(k\wedge n)}{n}\asymp \pth{\frac{k }{n}}^2 + \frac{\log^2 k}{n} $, which turns out to be the sharp characterization.

Comparing \prettyref{eq:main} and \prettyref{eq:Rplug-rate}, we reach the following verdict on the plug-in estimator: Empirical entropy is rate-optimal, \ie, achieving a constant factor of the minimax risk, if and only if we are in the ``data-rich'' regime $n = \Omega(\frac{k^2}{\log^2 k})$. In the ``data-starved'' regime of $n = o\big(\frac{k^2}{\log^2 k}\big)$, empirical entropy is strictly rate-suboptimal.

\subsection{Previous results}

Below we give a concise overview of the previous results on entropy estimation.
There also exists a vast amount of literature on estimating (differential) entropy on continuous alphabets which is outside the scope of the present paper
%Universal estimation of information measures . Mutual info, divergence, also continuous alphabet 
(see the survey \cite{WKV09b} and the references therein).

\paragraph{Fixed alphabet}
%\nb{Give a comprehensive review of previous work pre-Valiant-Valiant.}

%\cite{Basharin59} and Miller-Madhow, Harris \cite{Harris75}, Jonathan Victor,  Grassberger, Antos-Kontoyiannis,  etc. 

For fixed distribution $P$ and $n\to \infty$, Antos and Kontoyiannis \cite{AK01} showed that the plug-in estimator is always consistent and the asymptotic variance of the plug-in estimator is obtained in \cite{Basharin59}. However, the convergence rate of the bias can be arbitrarily slow on a possibly infinite alphabet.
The asymptotic expansion of the bias is obtained in, \eg, \cite{Miller55, Harris75}:
\begin{equation}
    \Expect[\Hplug(N)] = H(P) - \frac{S(P)-1}{2n} + \frac{1}{12n^2} \pth{1 - \sum_{i=1}^k\frac{1}{p_i}} + O(n^{-3}),
    \label{eq:Hplug-bias}
\end{equation}
where $S(P) = \sum_i \indc{p_i > 0}$ denote the support size. This inspired various types of bias reduction to the plug-in estimator, such as the Miller-Madow estimator \cite{Miller55}:
\begin{equation}
        \hat H_{\rm MM} = \Hplug + \frac{\hat S-1}{2n}
        \label{eq:miller-madow}
\end{equation}
where $\hat S$ is the number of observed distinct symbols.

%People reliazed the main problem of empirical entropy (plug-in estimator) is the bias, which is provably inconsistent unless the sample size far exceeds the alphabet size. Variance is always small.

\paragraph{Large alphabet}
%Estimation and compression probability distribution on large alphabets: Orlitsky et al., Pattern ML (PML) \cite{OSZ04}. 

%Batu, Rubinfeld, etc.
It is well-known that to estimate the distribution $P$ itself, say, with total variation loss at most a small constant, we need at least $\Theta(k)$ samples (see, \eg, \cite{BFSS02}). However, to estimate the entropy $H(P)$ which is a scalar function, it is unclear from first principles whether $n=\Theta(k)$ is necessary. This intuition and the inadequacy of plug-in estimator have already been noted by Dobrushin \cite{Dobrushin58}, who wrote:
\begin{quotation}
\emph{...This method (empirical entropy) is very laborious if $m$, the number of values of the random
variable is large, since in this case most of the probabilities $p_i$ are small and to determine each of
them we need a large sample of length $N$, 
which leads to a lot of work. However, it is natural to
expect that in principle the problem of calculating the single characteristic $H$ of the distribution
$(p_1,\ldots, p_m)$ is simpler than calculating the $m$-dimensional vector $(p_1, \ldots, p_m)$, and that therefore one ought to seek a solution of the problem by a method which does not require reducing the first and simpler problem to the second and more complicated problem.}
\end{quotation}

Using non-constructive arguments, Paninski first proved that it is possible to consistently estimate the entropy using \emph{sublinear} sample size, \ie, there exists $n_k=o(k)$, such that $R^*(k,n_k)\to 0$ as $k\to \infty$ \cite{Paninski04}. Valiant proved that no consistent estimator exists, \ie, $R^*(k,n_k) \gtrsim 1$ if $n \lesssim \frac{k}{\exp(\sqrt{\log k})}$ \cite{Valiant08}.
The sharp scaling of the minimal sample size of consistent estimation is shown to be $\frac{k}{\log k}$ in the breakthrough results of Valiant and Valiant \cite{VV10,VV11}. 
However, the optimal sample size as a function of alphabet size $k$ and estimation error $\epsilon$ has not been completely resolved. Indeed, an estimator based on linear programming is shown to achieve an additive error of $\epsilon$ using $\frac{k}{\epsilon^2 \log k}$ samples \cite[Theorem 1]{VV13}, while $\frac{k}{\epsilon \log k}$ samples are shown to be necessary \cite[Corollary 10]{VV10}. 
%$ \frac{1}{\sqrt{c}} $ when $ n=c\frac{k}{\log k}. $ 
%That is not tight in view of \cite[Corollary 10]{VV11} where the lower bound there is $ \frac{1}{c} $. 
This gap is partially amended in \cite{VV11-focs} by a different estimator, which requires $ \frac{k}{\epsilon \log k} $ samples but only valid when $ \epsilon>k^{-0.03} $.
\prettyref{thm:main} generalizes their result by characterizing the full minimax rate and the sharp sample complexity is given by \prettyref{eq:sample-complexity}.
%generalizing the results above.

%Wager-Viswanath-Kulkarni: restriction on distribution (within constant factors of $\frac{1}{n}$) -- what is the goal there? 

We briefly discuss the difference between the lower bound strategy of \cite{VV10} and ours. Since the entropy is a permutation-invariant functional of the distribution, a sufficient statistic for entropy estimation is the histogram of the histogram $N$:
%a further summary of the sufficient statistics $N$:
\begin{equation}
 h_i=\sum_{j=1}^k \indc{N_j=i}, \quad i \in [n],
    \label{eq:fp}
\end{equation}
also known as \emph{histogram order statistics} \cite{Paninski03}, \emph{profile} \cite{OSZ04}, or \emph{fingerprint} \cite{VV10}, which is the number of symbols that appear exactly $i$ times in the sample.  A canonical approach to obtain minimax lower bounds for functional estimation is Le Cam's two-point argument \cite[Chapter 2]{Lecam86}, \ie, finding two distributions which have very different entropy but induce almost the same distribution for the sufficient statistics, in this case, the histogram $N_1^k$ or the fingerprints $h_1^n$, both of which have non-product distributions.
 A frequently used technique 
 to reduce dependence is 
 \emph{Poisson sampling} (see \prettyref{sec:poi}), where we relax the fixed sample size to a Poisson random variable with mean $n$. This does not change the statistical nature of the problem due to the exponential concentration of the Poisson distribution near its mean. Under the Poisson sampling model, the sufficient statistics $N_1,\ldots,N_k$ are independent Poissons with mean $np_i$; however, the entries of the fingerprint remain highly dependent.
 To contend with the difficulty of computing statistical distance between high-dimensional distributions with dependent entries, the major tool in \cite{VV10} is a new CLT for approximating the fingerprint distribution by quantized Gaussian distribution, which are parameterized by the mean and covariance matrices and hence more tractable. This turns out to improve the lower bound in \cite{Valiant08} obtained using Poisson approximation. 
 
 In contrast, in this paper we shall not deal with the fingerprint directly, but rather use the original sufficient statistics $N_1^k$ due to their independence endowed by the Poissonized sampling. Our lower bound relies on choosing two random distributions (priors) with almost \iid entries which effectively reduces the problem to one dimension, thus circumventing the hurdle of dealing with high-dimensional non-product distributions. The main intuition is that a random vector with \iid entries drawn from a positive unit-mean distribution is not exactly but \emph{sufficiently close} to a probability vector due to the law of large numbers, so that effectively it can be used as a prior in the minimax lower bound.
%  modulo certain soft arguments.

While the focus of the present paper is estimating the entropy under the additive error criterion, approximating the entropy multiplicatively has been considered in \cite{BDKR05}. It is clear that in general approximating the entropy within a constant factor is impossible with any finite sample size  (consider Bernoulli distributions with parameter $1$ and $1-2^{-n}$, which are not distinguishable with $n$ samples); 
nevertheless, when the entropy is large enough, \ie, $H(P) \gtrsim \gamma/\eta$, 
it is possible to approximate the entropy within a multiplicative factor of $\gamma$ using $n \lesssim k^{(1+\eta)/\gamma^2} \log k$ number of samples (\cite[Theorem 2]{BDKR05}).

%Canonical approaches on minimax lower bounds in statistical decision theory include
%Le Cam's two-point argument \cite[Chapter 2]{Lecam86} and information-theoretic methods involving Kolmogorov's metric entropy and Fano's inequality \cite{IKbook,YB99}. 

%they need to construct two distributions that are close to uniform distributions over all alphabet and half alphabet respectively in term of earth move distance. Close earth move distance implies the close entropy separation, therefore the two distributions have constant entropy separation. They also need to show that the law of fingerprints are indistinguishable. The fingerprints are generalized multinomial distributed, their new CLTs show that multinomial distribution are close to discretized Gaussian distribution. They show that the two fingerprints have close mean, which further implies close covariance. Therefore triangle inequality yields that the two fingerprints are indistinguishable.  

%As discussed as an open question at the end of \cite[Section 8]{Paninski03}, the minimax rate of the MSE of entropy estimation remains unknown.

%%%%%%%%% more related results...

%On a related note, estimating entrpoy within multiplicative factors are investigated in \cite{BDKR05}. Connections to statistical physics approximation are explored in \cite{Vontobel14}.

\subsection{Best polynomial approximation}

The proof of both the upper and the lower bound in \prettyref{thm:main} relies on the apparatus of \emph{best polynomial approximation}.
Our inspiration comes from previous work on functional estimation in Gaussian mean models \cite{LNS99,CL11}.  
Nemirovski (credited in \cite{INK87}) pioneered the use of polynomial approximation in functional estimation and showed that unbiased estimators for the truncated Taylor series of the smooth functionals is asymptotically efficient.
This strategy is generalized to non-smooth functionals in  \cite{LNS99} using best polynomial approximation and in \cite{CL11} for estimating the $\ell_1$-norm in Gaussian mean model. 
%Lower bound: Le Cam's argument with two priors (also known as fuzzy hypothesis testing in \cite{Tsybakov09}). 

On the constructive side, the main idea is to trade bias with variance. Under the \iid sampling model, it is easy to show (see, \eg, \cite[Proposition 8]{Paninski03}) that to estimate a functional $f(P)$ using $n$ samples, an unbiased estimator exists if and only if $f(P)$ is a polynomial in $P$ of degree at most $n$. Similarly, under Poisson sample model, $f(P)$ admits an unbiased estimator if and only if $f$ is real analytic.
Consequently, there exists no unbiased entropy estimator with or without Poissonized sampling.
%In many statistical models, unbiased estimator exists for monomials and hence for any polynomials, but not for entropy involving logarithms (see \cite[Proposition 8, p. 1236]{Paninski03}).
 Therefore, a natural idea is to approximate the entropy functional by polynomials which enjoy unbiased estimation, and reduce the bias to at most the uniform approximation error. The choice of the degree aims to strike a good bias-variance balance. 
In fact, the use of polynomial approximation in entropy estimation is not new. In \cite{Vinck12}, the authors considered a truncated Taylor expansion of $\log x$ at $x=1$ which admits an unbiased estimator, and proposed to estimate the remainder term using Bayesian techniques; however, no risk bound is given for this scheme.
Paninski also studied how to use approximation by Bernstein polynomials to reduce the bias of the plug-in estimators \cite{Paninski03}, which forms the basis for proving the existence of consistent estimators with sublinear sample complexity in \cite{Paninski04}.

Shortly before we posted this paper to arXiv, we learned that Jiao et al. \cite{JVHW15} independently used the idea of best polynomial approximation in the upper bound of estimating Shannon entropy and power sums with a slightly different estimator which also achieves the minimax rate.
    For more recent results on estimating Shannon entropy, support size, R\'enyi entropy and other distributional functionals on large alphabets, see \cite{JVHW14,AOST15,WY2015,HJW15,HJW15adaptive}.
    In particular, \cite{HJW15adaptive} sharpened \prettyref{thm:main} by giving a constant-factor characterization of the minimax risk in the regime of $ n\lesssim \frac{k}{\log k} $ using similar techniques developed in this paper.

While the use of best polynomial approximation on the constructive side is admittedly natural, the fact that it also arises in the optimal lower bound is perhaps surprising.
As carried out in \cite{LNS99,CL11}, the strategy is to choose two priors with matching moments up to a certain degree, which ensures the impossibility to test. The minimax lower bound is then given by the maximal separation in the expected functional values subject to the moment matching condition. This problem is the \emph{dual} of best polynomial approximation in the optimization sense (see \prettyref{app:moments} for a self-contained account). For entropy estimation, this approach yields the optimal minimax lower bound, although the argument is considerably more involved due to the extra constraint imposed by probability vectors.
% on the mean of the prior.

%Best polynomial approximation used in \cite{Paninski03} which takes an approximation-theoretic perspective. Recover classical results on asymptotic bias of the empirical entropy. But in \cite{Paninski03} the attention is still restricted to the scope of bias correction for the empirical entropy estimator, in fact, used to work toword the best possible bias correction.
%However,
%
% He looked at approximation of $p \mapsto \phi(p) \triangleq p \log p$. In constrast, we use $p \mapsto \log p$. In fact, in \cite[Section 6]{Paninski03} Paninski has realized that uniform approximation is too pessimistic, since when $p$ is smaller we can tolerate a larger approximation error. Therefore a weighted apprximation error criterion of $\phi$ is used where the weight function is proportional to $\frac{1}{p}$ when $p$ is small. This already gets us close to uniform approximation of $\log$. 
%In fact, using the normalization of the probability mass itself, we show that uniform approximation of the log function is the right thing to look at.
% 
% 
%  also degree is sample size $n$. this is way too big to trade bias and variance. Note that we give a precise analysis of the best bias-variance tradeoff, while the estimator $\hat{H}_{BUB}$ is motivated by optimizing the best upper bound on bias and variance.

%\subsection{Organization}

\paragraph{Notations} 
Throughout the paper all logarithms are with respect to the natural base and the entropy is measured in nats.
%$X_i \inddistr P_i$ denotes ...
Let $\Poi(\lambda)$ denote the Poisson distribution with mean $\lambda$ whose probability mass function is $\poi(\lambda,j) \triangleq \frac{\lambda^j e^{-\lambda}}{j!}, j \in \integers_+$.
Given a distribution $P$, its $n$-fold product is denoted by $P^{\otimes n}$.
For a parametrized family of distributions $\{P_\theta\}$ and a prior $\pi$, the mixture is denoted by $\expects{P_\theta}{\pi} = \int P_{\theta} \pi(\diff \theta)$. In particular, $\expect{\Poi\pth{U}}$ denotes the Poisson mixture with respect to the distribution of a positive random variable $U$.
The total variation and Kullback-Leibler (KL) divergence between probability measures $P$ and $Q$ are respectively given by $\TV(P,Q) = \frac{1}{2} \int |\diff P-\diff Q|$ and $D(P\|Q) =  \int \diff P \log \frac{\diff P}{\diff Q}$.
Let $\Bern(p)$ denote the Bernoulli distribution with mean $p$.

%Throughout the paper, we use $c,c',c_0,C$ to denote generic absolute positive constants, though the actual value may vary at different occasions.

%\input{setup}
\section{Poisson sampling}
\label{sec:poi}
The multinomial distribution of the sufficient statistic $N=(N_1,\ldots,N_k)$ is difficult to analyze because of the dependency.
A commonly used technique is the so-called \emph{Poisson sampling},
%(see, \eg, \cite{Szpankowski01})
 where we relax the sample size $n$ from being deterministic to a Poisson random variable $n'$ with mean $n$. Under this model, 
 we first draw the sample size $ n'\sim \Poi(n) $, then draw $ n' $ \iid\ samples from the distribution $P$. 
The main benefit is that now the sufficient statistics $ N_i\inddistr \Poi(np_i) $ are independent, which significantly simplifies the analysis.

Analogous to the minimax risk \prettyref{eq:Rkn}, we define its counterpart under the Poisson sampling model:
\begin{equation}
    \tilde{R}^*(k,n) \triangleq \inf_{\hat{H}}\sup_{P \in \calM_k}\Expect( \hat{H}(N)-H(P) )^2,    
    \label{eq:Rknt}
\end{equation}
where $ N_i\inddistr \Poi(np_i) $ for $ i=1, \dots, k $. 
%Note that by definition we have
%\begin{equation}
%\tilde{R}^*(k,n) \leq \Expect_{n'\sim\Poi(n)}[ R^*(k,n')] = \sum_{m \geq 0} R^*(k,m) \poi(n,m).
%    \label{eq:poisson-risk}
%\end{equation}
%where $n'\sim \Poi(n)$.
In view of the exponential tail of Poisson distributions, the Poissonized sample size is concentrated near its mean $ n $ with high probability, which guarantees that the minimax risk under Poisson sampling is provably close to that with fixed sample size.
Indeed, the inequality
\begin{equation}
    \tilde{R}^*(k,2 n) - \exp(-n/4)\log^2 k  \leq R^*(k,n) \leq 2\tilde{R}^*(k,n/2)
    \label{eq:RRt}
\end{equation}
allows us to focus on the risk of the Poisson model (see \prettyref{app:poisson} for a proof).

%\prettyref{thm:poisson} rigorously shows Poisson sampling will only introduce a constant factor on the sample complexity of original problem.  

\section{Minimax lower bound}
\label{sec:lb}

In this section we give converse results for entropy estimation and prove the lower bound part of \prettyref{thm:main}.
It suffices to show that the minimax risk is lower bounded by the two terms in \prettyref{eq:main} separately. This follows from combining Propositions \ref{prop:lb1} and \ref{prop:lb2} below.

%\begin{theorem}
%%Let $k,n\in\naturals$. 
%For any $k \geq 2$ and $n \geq 1$, we have
%\begin{equation}
%   R^*(k,n) \gtrsim \pth{\frac{k }{n \log k}}^2 + \frac{\log^2 k}{n}
%   \label{eq:lb}
%\end{equation} 
%   \label{thm:lb}
%\end{theorem}

\begin{prop}
    For all $k,n\in \naturals$,
    \begin{equation}
        R^*(k,n) \gtrsim  \frac{\log^2 k}{n}.   
        \label{eq:lb1}
    \end{equation}
    \label{prop:lb1}
\end{prop}

\begin{prop}
%    For all $k,n\in \naturals$,
For all $k,n\in \naturals$,
    \begin{equation}
        R^*(k,n) \gtrsim \pth{\frac{k}{n \log k}}^2\vee 1.
        \label{eq:lb2}
    \end{equation}
    \label{prop:lb2}
\end{prop}

 \prettyref{prop:lb1}, proved in Appendix \ref{sec:pf-lb1}, follows from a simple application of Le Cam's \emph{two-point method}:
%, which involves a binary hypothesis testing argument:
 If two input distributions $P$ and $Q$ are sufficiently close such that it is impossible to reliably distinguish between them using $n$ samples with error probability less than, say, $\frac{1}{2}$, 
% for which a sufficient condition is that $nD(P\|Q)$ is bounded, 
 then any estimator suffers a quadratic risk proportional to the separation of the functional values $|H(P)-H(Q)|^2$.

The remainder of this section is devoted to outlining the broad strokes for proving \prettyref{prop:lb2}. 
The proofs as well as the intermediate results are elaborated in Appendix \ref{sec:pf-lb}. 
Since it can be shown that the best lower bound provided by the two-point method is $\frac{\log^2 k}{n}$ (see \prettyref{rmk:HD}), 
proving \prettyref{eq:lb2} requires more powerful techniques.
%one needs to go beyond the two-point method. 
To this end, 
we use a generalized version of Le Cam's method involving two \emph{composite} hypotheses (also known as fuzzy hypothesis testing in \cite{Tsybakov09}):
\begin{equation}
    H_0: H(P) \leq t \quad \text{versus} \quad H_1: H(P) \geq t+ d,
    \label{eq:compHT}
\end{equation}
which is more general than the two-point argument using only simple hypothesis testing. Similarly, if we can establish that no test can distinguish \prettyref{eq:compHT} reliably, then we obtain a lower bound for the quadratic risk on the order of $d^2$.
By the minimax theorem, the optimal probability of error for the composite hypotheses test is given by the Bayesian version with respect to the least favorable priors. For  \prettyref{eq:compHT} we need to choose a pair of priors, 
%or equivalently, a pair of random distributions, 
%supported on the entropy sublevel and superlevel sets in
which, in this case, are distributions on the probability simplex $\calM_k$, to ensure that the entropy values are separated. 
%Morever, the averaging with respect to the priors reduces the total variation distance and makes the joint distribution of the samples indistinguishable.

\subsection{Construction of the priors}
\label{sec:prior}
The main idea for constructing the priors is as follows: First of all, the symmetry of the entropy functional implies that the least favorable prior must be permutation-invariant. 
This inspires us to use the following \emph{\iid construction}. For conciseness, we focus on the case of $n \asymp \frac{k}{\log k}$ for now and our goal is to obtain an $\Omega(1)$ lower bound.
 Let $U$ be a $\reals_+$-valued random variable with unit mean. 
%We can think of $U$ as a ``mother of distribution''.
Consider the random vector 
 \[
\sfP=  \frac{1}{k} (U_1,\ldots,U_k),
\] 
consisting of \iid copies of $U$. Note that $\sfP$ itself is \emph{not} a probability distribution; however, the key observation is that, since $\Expect[U]=1$, as long as the variance of $ U $ is not too large, the weak law of large numbers ensures that $\sfP$ is \emph{approximately} a probability vector. 
%Indeed, $|\frac{1}{k} \sum U_i - 1|$ does not exceed $\sqrt{\frac{1}{k} \var U}$ with constant probability. 
 Using a conditioning arguments we can show that the distribution of $\sfP$ can effectively serve as a prior. 
 To gain more intuitions, note that, for example, a deterministic $U = 1$ generates a uniform distribution over $[k]$, while a binary $U \sim \frac{1}{2} (\delta_0+\delta_2)$ generates 
a uniform distribution over roughly half the alphabet with the support set uniformly chosen at random.
From this viewpoint, the CDF of the random variable $\frac{U}{k}$ plays the role of the \emph{histogram of the distribution} $\sfP$, which is the central object in the Valiant-Valiant lower bound construction (see \cite[Definition 3]{VV10}).

Next we outline the main ingredients in implementing Le Cam's method:
\begin{enumerate}
    \item \emph{Functional value separation}: Define $\phi(x) \triangleq x \log \frac{1}{x}$. Note that
    \[
H(\sfP)= \sum_{i=1}^k\phi\pth{\frac{U_i}{k}} =\frac{1}{k}\sum_{i=1}^k\phi(U_i)+\frac{\log k}{k}\sum_{i=1}^kU_i,
\]
which concentrates near its mean $ \expect{H(\sfP)}=\expect{\phi(U)}+\expect{U}\log k$ by law of large numbers. Therefore, given another random variable $U'$ with unit mean, we can obtain $\sfP'$ similarly using \iid copies of $U'$. Then with high probability, $H(\sfP)$ and $H(\sfP')$ are separated by the difference of their mean values, namely,
%\begin{equation}
%|\expect{H(\sfP)}-\expect{H(\sfP')}| = |\expect{\phi(U)} - \expect{\phi(U')}|, 
%   \label{eq:hu1}
%\end{equation}
\[
\expect{H(\sfP)}-\expect{H(\sfP')} = \expect{\phi(U)} - \expect{\phi(U')},
\]
which we aim to maximize.

\item \emph{Indistinguishably}: 
Note that given $P$, the sufficient statistics satisfy $N_i\inddistr \Poi(n p_i)$. Therefore, if $P$ is drawn from the distribution of $\sfP$, then $N=(N_1,\ldots,N_k)$ are \iid distributed according the \emph{Poisson mixture} $\Expect[\Poi(\frac{n}{k} U)]$. Similarly, if $P$ is drawn from the prior of $\sfP'$, then $N$ is distributed according to $(\Expect[\Poi(\frac{n}{k} U')])^{\otimes k}$. 
To establish the impossibility of testing, we need the total variation distance between the two $k$-fold product distributions to be strictly bounded away from one, for which a sufficient condition is
\begin{equation}
    \TV(\Expect[\Poi(n U/k)],\Expect[\Poi(n U'/k)]) \leq c/k
    \label{eq:hu2}
\end{equation}
for some $c<1$.

%In view of the fact that $\TV(Q^{\times})$, 
\end{enumerate}
To conclude, we see that the \iid construction fully exploits the independence blessed by the Poisson sampling, thereby reducing the problem to \emph{one dimension}. This allows us to sidestep the difficulty encountered in \cite{VV10} when dealing with fingerprints which are high-dimensional random vectors with dependent entries.
%: introduce  effectively single-letterize the problem boils down to one dimension.

What remains is the following scalar problem: choose $U,U'$ to maximize $|\expect{\phi(U)} - \expect{\phi(U')}|$ subject to the constraint \prettyref{eq:hu2}.
A commonly used proxy for bounding the total variation distance is \emph{moment matching}, \ie, $\expect{U^j} = \expect{U'^j}$ for all $j = 1,\ldots,L$. Together with $L_\infty$-norm constraints, a sufficient large degree $L$ ensures the total variation bound \prettyref{eq:hu2}.
Combining the above steps, our lower bound is proportional to the value of the following convex optimization problem (in fact, infinite-dimensional linear programming over probability measures):
\begin{equation}
    \begin{aligned}
     \calF_L(\lambda) \triangleq  \sup & ~ \expect{\phi(U)} - \expect{\phi(U')}  \\
        \text{s.t.}     
        & ~ \expect{U} = \expect{U'}=1 \\
        & ~ \expect{U^j} = \expect{U'^j}, \quad j = 1,\ldots,L, \\
        & ~ U,U' \in 
%        \nb{?\sth{0}\cup[1,\lambda]} 
[0, \lambda]
    \end{aligned}
    \label{eq:FL}
\end{equation}
%where $L \asymp \log k$ and $\lambda \asymp \frac{k L}{n}$.
for some appropriately chosen $L \in \naturals$ and $\lambda > 1$ depending on $n$ and $k$. 
%This problem is actually equivalent to another problem without restriction on the mean: (see \prettyref{app:opt} for a short proof.) 
%\begin{equation}
%    \begin{aligned}
%     \calF_L(\lambda) \triangleq  \sup & ~ \expect{\log\frac{1}{X}} - \expect{\frac{1}{X'}}  \\
%        \text{s.t.}     
%        & ~ \expect{X^j} = \expect{X'^j}, \quad j = 1,\ldots,L+1, \\
%        & ~ X,X' \in [\eta,1]
%    \end{aligned}
%\end{equation}
%where $ \eta=\lambda^{-1} $.

Finally, we connect the optimization problem \prettyref{eq:FL} to the machinery of \emph{best polynomial approximation}: Denote by $\mathcal{P}_L$ the set of polynomials of degree $L$ and
\begin{equation}
E_L(f,I)    \triangleq \inf_{p\in\mathcal{P}_L}\sup_{x\in I }|f(x)-p(x)|,
    \label{eq:EL}
\end{equation}
which is the best uniform approximation error of a function $f$ over a finite interval $I$ by polynomials of degree $L$. We prove that 
\begin{equation}
    \calF_L(\lambda) \geq 2 E_L(\log, [1/\lambda,1]).
    \label{eq:FLEL}
\end{equation}
Due to the singularity of the logarithm at zero, the approximation error can be made bounded away from zero if $\lambda$ grows \emph{quadratically} with the degree $L$ (see \prettyref{app:error}). 
Choosing $L \asymp \log k$ and $\lambda \asymp \log^2 k$ leads to the impossibility of consistent estimation for $n \asymp \frac{k}{\log k}$.
For $n \gg \frac{k}{\log k}$, the lower bound for the quadratic risk follows from relaxing the unit-mean constraint in \prettyref{eq:FL} to $\expect{U} = \expect{U'} \leq 1$ and a simple scaling argument. We refer to the proofs in Appendix \ref{sec:pf-lb} for details.
Analogous construction of priors and proof techniques have subsequently been used in \cite{JVHW15} to obtain sharp minimax lower bound for estimating the power sum in which case the $\log p$ function is replaced by $p^\alpha$.

%It should be remarked that it is well-known that moment matching problem is closed related to best polynomial approximation (see \prettyref{app:moments} for a self-contained account). However, \prettyref{eq:FL} requires unit mean while in general we have no control over the mean of the distributions. In order to achieve a prescribed mean, we turn to the approximation of $\log x$ then relate back to $x \log x$.

%Main steps
%\begin{enumerate}
%   \item Moment matching and $L_\infty$ norm constraint implies the corresponding Poisson mixture.
%   \item It is well-known that moment matching is closed related to best polynomial approximation (see \prettyref{app:moments} for a self-contained account).
%   However, we have no control over the mean of the distributions. In order to achieve unit mean, we turn to the approximation of $\log x$ then relate back to $x \log x$.
%\end{enumerate}

%\newpage
\section{Optimal estimator via best polynomial approximation}
    \label{sec:upper}
    
%    \nb{Validate Poisson Sampling} 
As observed in various previous results as well as suggested by the minimax lower bound in \prettyref{sec:lb}, the major difficulty of entropy estimation 
%when the sample size does not far exceed the alphabet size. 
lies in the bias due to insufficient samples.
Recall that the entropy is given by $H(P)=\sum \phi(p_i)$, where $\phi(x) = x \log \frac{1}{x}$. 
It is easy to see that the expectation of any estimator $T: [k]^n \to \reals_+$ is a polynomial of the underlying distribution $P$ and, consequently, no unbiased estimator for the entropy exists (see, \eg, \cite[Proposition 8]{Paninski03}). This observation inspired us to approximate $\phi$ by a polynomial of degree $L$, say $g_L$, for which we pay a price in bias as the approximation error but yield the benefit of zero bias. While the approximation error clearly decreases with the degree $L$, it is not unexpected that the variance of the unbiased estimator for $g_L(p_i)$ increases with $L$ as well as the corresponding mass $p_i$. Therefore we only apply the polynomial approximation scheme to small $p_i$ and directly use the plug-in estimator for large $p_i$, since the signal-to-noise ratio is sufficiently large.

    Next we describe the estimator in detail. In view of the relationship \prettyref{eq:RRt} between the risks with fixed and Poisson sample size, we shall assume the Poisson sampling model to simplify the analysis, where we first draw $ n'\sim \Poi(2n) $ and then draw $ n' $ \iid samples $ X=(X_1,\dots,X_{n'}) $ from $P$.
    We split the samples equally and use the first half for selecting to use either the polynomial estimator or the plug-in estimator and the second half for estimation.
    Specifically, for each sample $ X_i $ we draw an independent fair coin $ B_i\iiddistr \Bern\pth{\frac{1}{2}} $. We split the samples $ X $ according to the value of $ B $ into two sets and count the samples in each set separately. That is, we define $ N=(N_1,\dots,N_k) $ and $ N'=(N_1',\dots,N_k') $ by
    \begin{align*}
%        & N_i=\sum_{j=1}^{n'}\indc{X_j=i}\indc{B_j=0},\\
%        & N_i'=\sum_{j=1}^{n'}\indc{X_j=i}\indc{B_j=1}.
  N_i=\sum_{j=1}^{n'}\indc{X_j=i}\indc{B_j=0},\quad  N_i'=\sum_{j=1}^{n'}\indc{X_j=i}\indc{B_j=1}.
    \end{align*}
    Then $ N$ and $N'$ are independent, where $ N_i,N_i'\iiddistr\Poi\pth{np_i} $.

    Let $c_0,c_1,c_2>0$ be constants to be specified. Let $L= \floor{c_0\log k}$.
    Denote the best polynomial of degree $ L $ to uniformly approximate $ x\log\frac{1}{x} $ on $ \qth{0,1} $ by 
    \begin{equation}
    p_L(x)=\sum_{m=0}^{L}a_mx^m .
        \label{eq:pL}
\end{equation} Through a change of variables, we see that the best polynomial of degree $ L $ to approximate $ x\log\frac{1}{x} $ on $ [0, \frac{c_1\log k}{n}]$ is 
    \[
%    P_L(x)\triangleq % \frac{c_1\log k}{n}p_L\pth{\frac{xn}{c_1\log k}}+x\log\frac{n}{c_1\log k}=
%    \sum_{m\neq 1}\frac{a_mn^{m-1}}{\pth{c_1\log k}^{m-1}}x^m+\pth{a_1+\log\frac{n}{c_1\log k}}x.
    P_L(x)\triangleq     \sum_{m=0}^L \frac{a_mn^{m-1}}{\pth{c_1\log k}^{m-1}}x^m+\pth{ \log\frac{n}{c_1\log k}} x.
    \]
    Define the factorial moment by $ (x)_m\triangleq\frac{x!}{(x-m)!} $, which gives an unbiased estimator for the monomials of the Poisson mean: $ \Expect[(X)_m]=\lambda^m $ where $ X\sim\Poi(\lambda) $. Consequently, the following polynomial of degree $L$
\begin{equation}
        g_L(N_i)
    \triangleq \frac{1}{n}\pth{\sum_{m=0}^L \frac{a_m}{\pth{c_1\log k}^{m-1}} (N_i)_{m}
    + \pth{ \log\frac{n}{c_1\log k}} N_i}
        \label{eq:gL}
\end{equation}    
    is an unbiased estimator for $ P_L(p_i) $. 

    Define a preliminary estimator of entropy $ H(P)=\sum_{i=1}^{k}\phi(p_i) $ by
    \begin{equation}
        \tilde{H}\triangleq\sum_{i=1}^{k}\pth{g_L(N_i)\indc{N_i'\le c_2\log k}
            +\pth{\phi\pth{\frac{N_i}{n}}+\frac{1}{2n}}\indc{N_i'>c_2\log k}},
            \label{eq:tH}
    \end{equation}
    where we apply the estimator from polynomial approximation if $ N_i'\le c_2\log k $ or the bias-corrected plug-in estimator otherwise (c.f. the asymptotic expansion \prettyref{eq:Hplug-bias} of the bias under the original sampling model). 
    In view of the fact that $0 \leq H(P) \leq \log k $ for any distribution $ P $ with alphabet size $ k $, we define our final estimator by:
    \begin{equation*}
        \hat{H}
        =(\tilde{H}\vee 0)  \wedge \log k,
    \end{equation*}
    Since \prettyref{eq:tH} can be expressed in terms of a linear combination of the fingerprints \prettyref{eq:fp} of the second sample and the coefficients can be pre-computed using fast best polynomial approximation algorithms (\eg, the Remez algorithm), it is clear that the estimator $\hat{H}$ can be computed in linear time in $n$.

    The next result, proved in Appendix \ref{sec:pf-ub} gives an upper bound on the above estimator under the Poisson sampling model, which, in view of the right inequality in \prettyref{eq:RRt} and \prettyref{prop:Rplug-rate}, implies the upper bound on the minimax risk $R^*(n,k)$ in \prettyref{thm:main}. 
    \begin{prop}
        \label{prop:err-rate}
        Assume that $\log n \leq  C \log k$ for some constant $C>0$. Then there exists $c_0,c_1,c_2$ depending on $C$ only, such that
        \[
        \sup_{P \in \calM_k} \Expect[(H(P)-\hat{H}(N))^2]
        \lesssim \pth{\frac{k}{n\log k}}^2+\frac{\log^2 k}{n},
        \]
        where $N=(N_1,\ldots,N_k) \inddistr \Poi(n p_i)$.
        % Moreover, if $\log n \lesssim \log k$, then
        % \[
        % R^*(k,n) \leq 
        % \]
    \end{prop}
    
    \begin{remark}
    \label{rmk:linear}
    The benefit of sample splitting is that we can first condition on the realization of $N'$ and treat the indicators in \prettyref{eq:tH} as deterministic, which has also been used in the entropy estimator in \cite{JVHW15}.
Although not ideal operationally or aesthetically, this is a frequently-used idea in statistics and learning to simplify the analysis (also known as sample cloning in the Gaussian model \cite{Nemirovski03,CL11})
%    (see, \eg, the aggregation estimator in \cite{CMW12}) 
    at the price of losing half of the sample thereby inflating the risk by a constant factor.
    It remains to be shown whether the optimality result in \prettyref{prop:err-rate} continues to hold if we can use the same sample in \prettyref{eq:tH} for both selection and estimation.

%    \nb{Does \cite[Appendix A]{Paninski03} say anything about the optimality of estimators linear in histogram of histograms?}
%     The question whether optimal estimator is linear is raised in \cite{VV13}. So our construction gives an optimal estimator which is linear. 

    Note that the estimator \prettyref{eq:tH} is \emph{linear} in the fingerprint of the second half of the sample.              
We also note that for estimating other distribution functionals, \eg, support size \cite{WY2015}, it is possible to circumvent sample splitting by directly using a linear estimator obtained from best polynomial approximation. Similar ideas can be used to construct entropy estimators which are linear in the fingerprints and minimax rate-optimal \cite{yang-msthesis}.     
        
\end{remark}
  
\begin{remark}
The estimator \prettyref{eq:tH} uses the polynomial approximation of $x\mapsto x\log \frac{1}{x}$ for those masses below $\frac{\log k}{n}$ and the bias-corrected plug-in estimator otherwise. In view of the fact that the lower bound in \prettyref{prop:lb2} is based on a pair of randomized distributions whose masses are below $\frac{\log k}{n}$ (except for possibly a fixed large mass at the last element), this suggests that the main difficulty of entropy estimation lies in those probabilities in the interval $[0,\frac{\log k}{n}]$, which are individually small but collectively contribute significantly to the entropy.     See \prettyref{rmk:prior} and the proof of \prettyref{prop:lb2} for details.
    \label{rmk:cutoff}
\end{remark}

% \begin{remark}
%     In \cite[Theorem 1]{VV13}, their estimator via linear programming provides quadratic risk of $ \frac{1}{c} $ when $ n=c\frac{k}{\log k}. $ That is not tight in view of \cite[Corollary 10]{VV11} where the lower bound on minimax quadratic risk is $ \frac{1}{c^2} $. While our estimator shows the optimal quadratic risk. 
%     \label{rmk:tight}
% \end{remark}
    
%    \input{numerical}
    \begin{remark}
    % \nbr{
    %     Adaptivity to the alphabet size and setting $g(0)=0$. Include a plot to explain that setting $g(0)=0$, although having no impact on minimax rate optimality, leads to underbias.}
%    The entropy estimator in \cite{JVHW15} is constructed using similar ideas of best polynomial approxiamtion of the $x \mapsto x \log x$ function and sample splitting, with the following distinctions. First, the indicator function in \prettyref{eq:tH} is replaced by a smooth cutoff function; this seems to neither improve the error rate nor help the empirical performance (see \prettyref{sec:num}).     
    The estimator in \prettyref{eq:tH} depends on the alphabet size $k$ only through its logarithm; therefore the dependence on the alphabet size is rather insensitive. In many applications such as neuroscience the discrete data are obtained from quantizing an analog source and $k$ is naturally determined by the quantization level \cite{SLSKB97}. Nevertheless, it is also desirable to obtain an optimal estimator that is adaptive to $k$. To this end, 
    we can replace all $\log k$ by $\log n$ and define the final estimator by $ \tilde{H}\vee 0 $. Moreover, we need to set $ g_L(0)=0 $ since the number of unseen symbols is unknown. 
    Following \cite{JVHW15}, we can simply let the constant term $ a_0$ of the approximating polynomial \prettyref{eq:pL} to zero and obtained the corresponding unbiased estimator \prettyref{eq:gL} through factorial moments, which satisfies $ g_L(0)=0 $ by construction.\footnote{Alternatively, we can directly set $ g_L(0)=0 $ 
   and use the original $g_L(j)$ in \prettyref{eq:gL} when $j \geq 1$. Then the bias becomes $ \sum_i(P_L(p_i)-\phi(p_i)-\prob{N_i=0}P_L(0)) $. In sublinear regime that $ n=o(k) $, we have $ \sum_i\prob{N_i=0}=\Theta(k) $ therefore this modified estimator also achieves the minimax rate.}
    The bias upper bound becomes $ \sum_i(P_L(p_i)-\phi(p_i)-P_L(0)) $ which is at most twice of original upper bound since $ P_L(0)\le \linf{P_L-\phi} $.
    \begin{figure}[ht]
        \centering
\includegraphics[width=\linewidth]{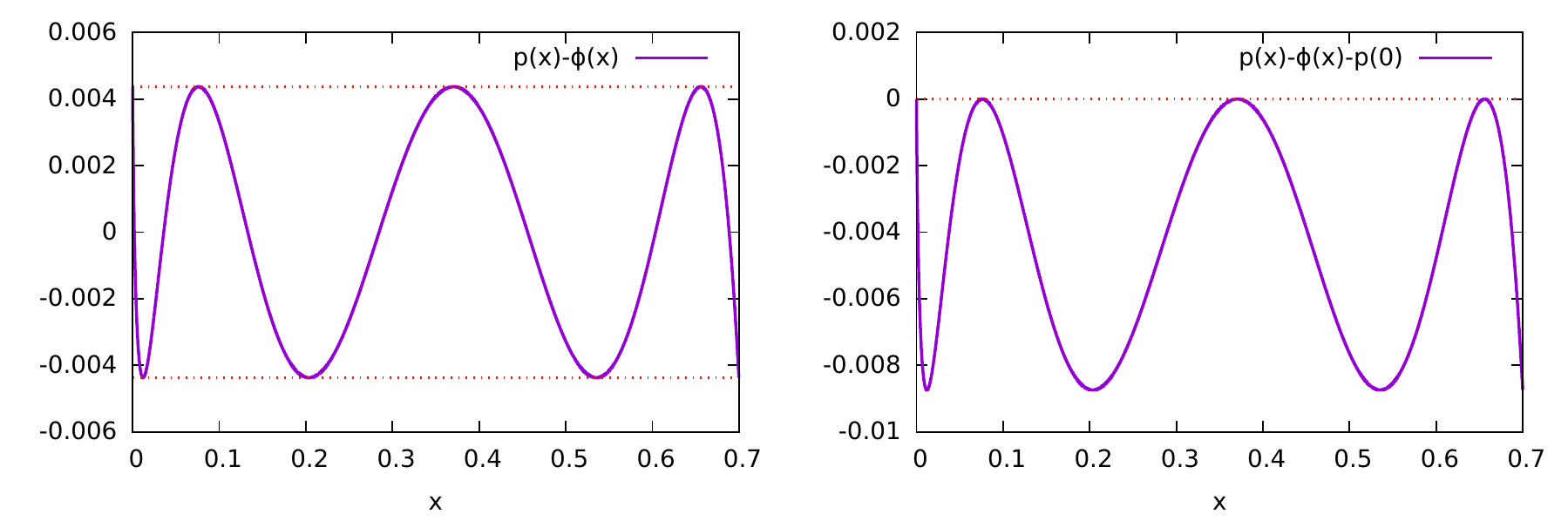}
        \caption{Bias of the degree-$ 6 $ polynomial estimator with and without the constant term.\label{fig:underbiased} }
    \end{figure}
    The minimax rate in \prettyref{prop:err-rate} continues to hold in the regime of $\frac{k}{\log k} \lesssim n \lesssim \frac{k^2}{\log^2 k}$, where the plug-in estimator fails to attain the minimax rate.
    In fact, $ P_L(0) $ is always strictly positive and coincides with the uniform approximation error (see \prettyref{app:end} for a short proof).
    Therefore removing the constant term leads to $ g_L(N_i) $ which is always underbiased as shown in \prettyref{fig:underbiased}.
    A better choice for adaptive estimation is to find the best polynomial 
    satisfying $ p_L(0)=0 $ that uniformly approximates $ \phi $.
    \label{rmk:adaptive}
\end{remark}

\section{Numerical experiments}
\label{sec:num}

In this section we compare the performance of our estimator described in \prettyref{sec:upper} to other estimators using synthetic data.\footnote{The C++ implementation of our estimator is available at \url{https://github.com/Albuso0/entropy}.}
Note that the coefficients of best polynomial to approximate $ \phi $ on $ [0,1] $ are independent of data so they can be pre-computed and tabulated to facilitate the computation in our estimation.
It is very efficient to apply Remez algorithm to obtain those coefficients which provably has linear convergence for all continuous functions % and quadratic convergence if the function to approximate satisfy some smoothness conditions
(see, \eg, \cite[Theorem 1.10]{petrushev2011rational}).
 % \nbr{So, the real question is: does $\phi$ satisfies this condition or not?} \nb{No.} \nbr{Why haven't you edited yet?}
Considering that the choice of the polynomial degree is logarithmic in the alphabet size, we pre-compute the coefficients up to degree $ 400 $ which suffices for practically all purposes.
In the implementation of our estimator we replace $ N_i' $ by $ N_i $ in \prettyref{eq:tH} without conducting sample splitting.
Though in the proof of theorems we are conservative about the constant parameters $ c_0,c_1,c_2 $, in experiments we observe that the performance of our estimator is in fact not sensitive to their value within the reasonable range.
In the subsequent experiments the parameters are fixed to be $ c_0=c_2=1.6, c_1=3.5 $.

\begin{figure}[!h]
    \centering
    \includegraphics[width=0.9\linewidth]{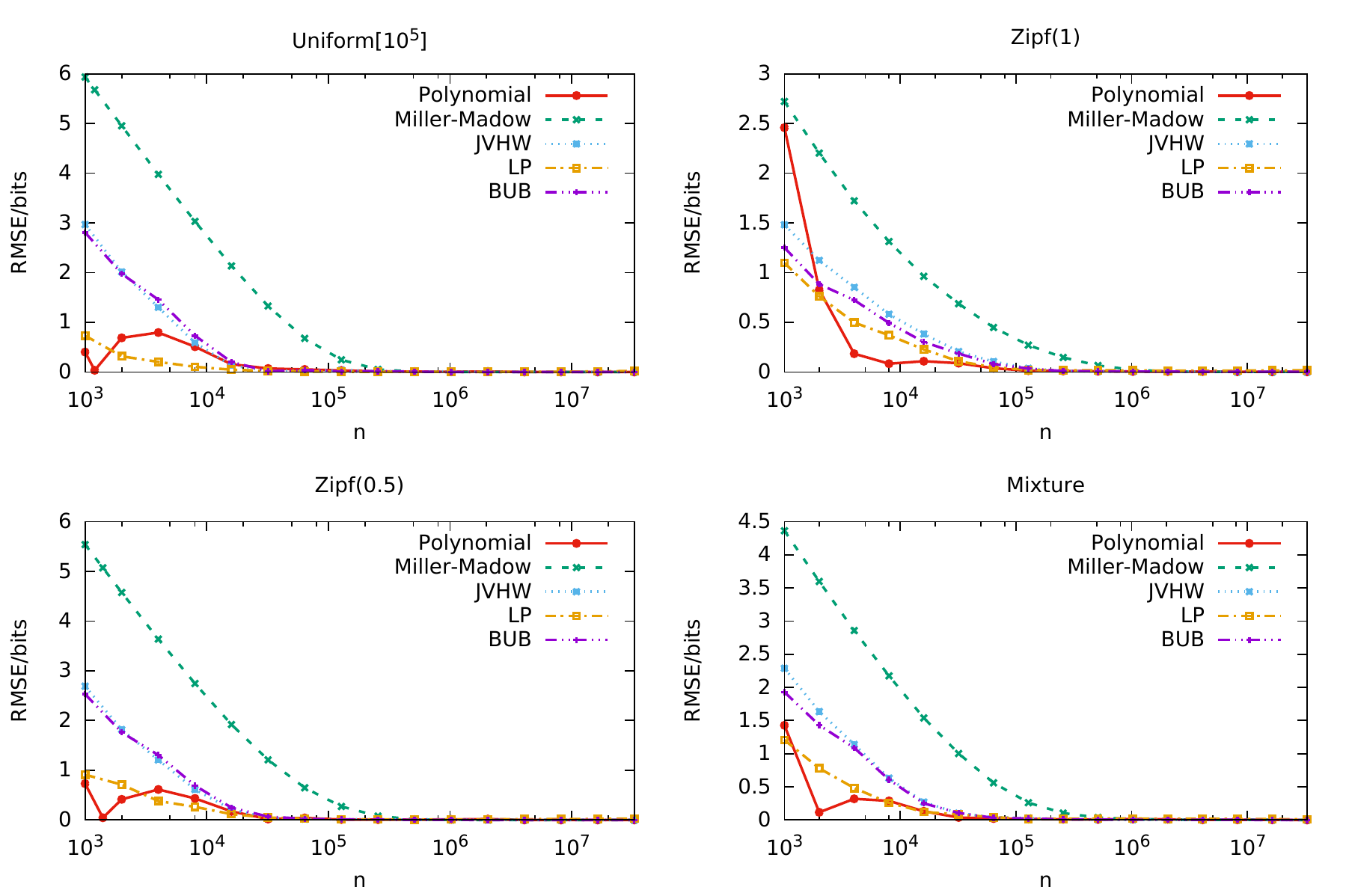}
    \caption{Performance comparison with sample size $n$ ranging from $ 10^3 $ to $3 \times 10^7 $. \label{fig:full} }
\end{figure}

\begin{figure}[!h]
    \centering
\includegraphics[width=0.9\linewidth]{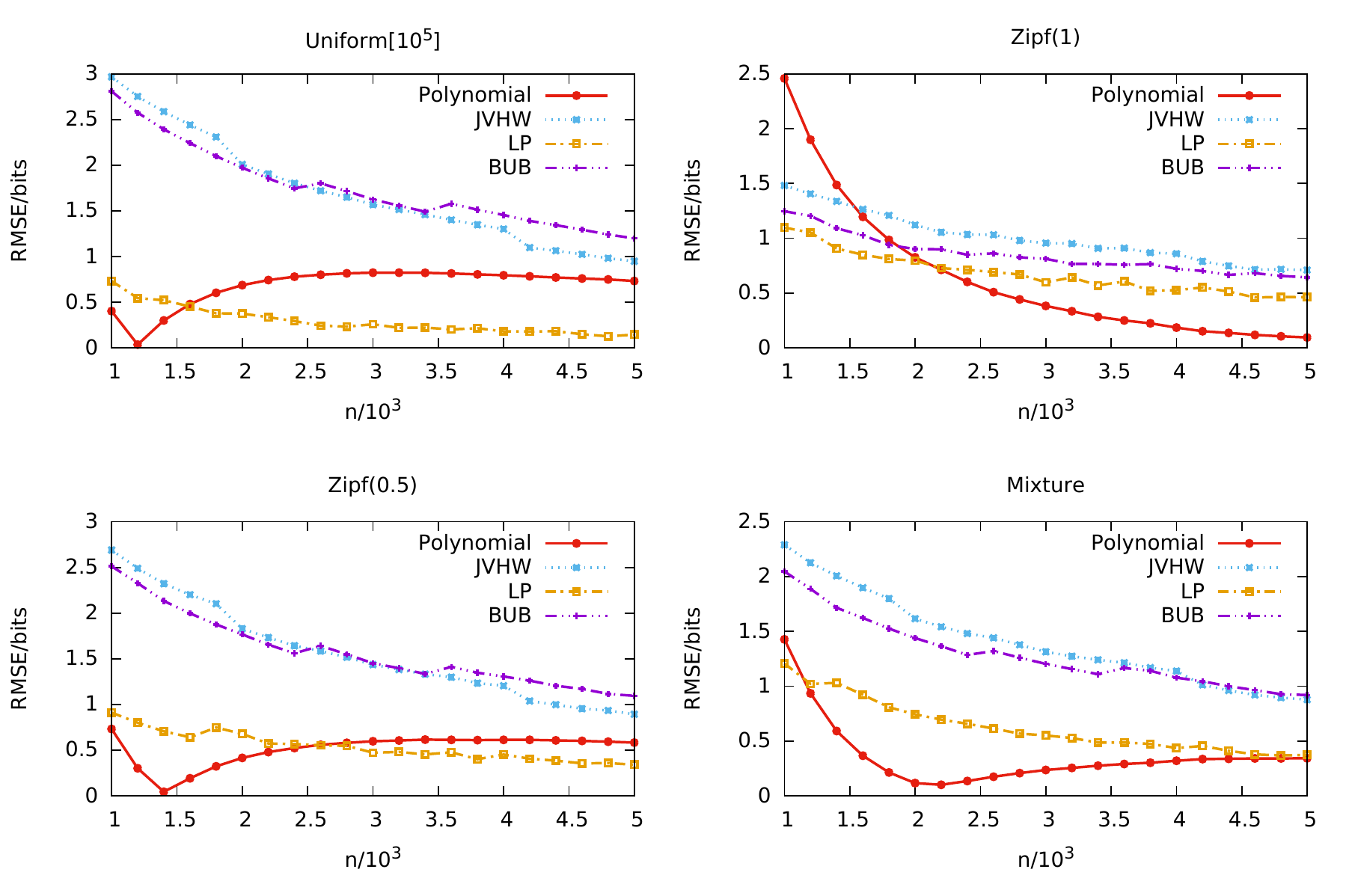}
    \caption{Performance comparison when sample size $n$ ranges from $ 1000 $ to $ 5000 $. \label{fig:scarce} }
\end{figure}

We generate data from four types of distributions over an alphabet of $ k=10^5 $ elements, namely, 
the uniform distribution with $ p_i=\frac{1}{k} $, Zipf distributions with $ p_i\propto i^{-\alpha} $ and $ \alpha$ being either $1$ or $0.5 $, and an ``even mixture'' of geometric distribution and Zipf distribution where  for the first half of the alphabet $ p_i \propto 1/i$ and  for the second half $ p_{i+k/2} \propto (1-\frac{2}{k})^{i-1} $, $ 1\le i\le \frac{k}{2} $.
%$ p_i=\frac{1/i}{2\sum_{j=1}^{k/2}1/j} $ and $ p_{k/2+i}=\frac{\frac{2}{k}(1-\frac{2}{k})^{i-1}}{2\sum_{j=1}^{k/2}\frac{2}{k}(1-\frac{2}{k})^{j-1}} $ for $ 1\le i\le \frac{k}{2} $.
Using parameters mentioned above, the approximating polynomial has degree $ 18 $, the parameter determining the approximation interval is $ c_1\log k=40 $, and the threshold to decide which estimator to use in \prettyref{eq:tH} is $ 18 $, namely, we apply the polynomial estimator $ g_L $ if a symbol appeared at most 18 times and the  bias-corrected plug-in estimator otherwise.
After obtaining the preliminary estimate $ \tilde{H} $ in \prettyref{eq:tH}, our final output is $ \tilde{H}\vee 0 $.\footnote{We can, as in \prettyref{prop:err-rate}, output $ (\tilde{H}\vee 0)\wedge \log k $, which yields a better performance. We elect not to do so for a stricter comparison.} 
%In order to demonstrate the overestimate of $ \tilde{H} $ is unobservable in the uniform distribution case.} 
Since the plug-in estimator suffers from severe bias when samples are scarce,
% (see, \eg, \cite{VV13,JVHW15}).
we forgo the comparison with it to save space in the figures and instead compare with its bias-corrected version, \ie, the Miller-Madow estimator \prettyref{eq:miller-madow}.
We also compare the performance with the linear programming estimator in \cite{VV13}, the best upper bound (BUB) estimator \cite{Paninski03}, and the estimator based on similar polynomial approximation techniques\footnote{The estimator in \cite{JVHW15} uses a smooth cutoff function in lieu of the indicator function in \prettyref{eq:tH}; this seems to improve neither the theoretical error bound nor the empirical performance.} proposed by \cite{JVHW15} 
 using their implementations with default parameters.
Our estimator is implemented in C++ which is much faster than those from \cite{VV13,JVHW15,Paninski03} implemented in MATLAB so the running time comparison is ignored.
We notice that the linear programming in \cite{VV13} is much slower than the polynomial estimator in \cite{JVHW15}, especially when the sample size becomes larger.

We compute the root mean squared error (RMSE) for each estimator over $ 50 $ trials.
The full performance comparison is shown in \prettyref{fig:full} where the sample size ranges from one percent to $ 300 $ folds of the alphabet size.
In \prettyref{fig:scarce} we further zoom into the more interesting regime of fewer samples with the sample size ranging from one to five percent of the alphabet size. In this regime our estimator as well as those from \cite{VV13,JVHW15,Paninski03} outperform the classical Miller-Madow estimator significantly; furthermore, our estimator performs better than those in \cite{JVHW15,Paninski03} in most cases tested and comparably with that in \cite{VV13}.
When the samples are abundant all estimators achieve very small error; however, it has been empirically observed in \cite{JVHW15} that the performance of linear programming starts to deteriorate when the sample size is very large, which is also observed in our experiments (see \cite{yang-msthesis}). The specific figures of that regime are ignored since the absolute errors are very small and the even the plug-in estimator without bias correction is accurate.
By \prettyref{eq:tH}, for large sample size our estimator tends to the Miller-Madow estimator when every symbol is observed many times.

\appendix
\section{A risk bound for the Poisson Sampling model}
\label{app:poisson}

%Since $0 \leq R^*(k,m) \leq \log^2 k$, in view of the fact that $m \mapsto R^*(k,m)$ is decreasing and applying Markov's inequality and 
%the Chernoff bound (see, \eg, \cite[Theorem 5.4]{MU06}, we have

%Here we prove the inequality \prettyref{eq:RRt} relating the minimax quadratic risk of the entropy estimation under the usual iid sampling model \prettyref{eq:Rkn} to that under the Poisson sampling model \prettyref{eq:Rknt}. Recall that $0 \leq R^*(k,m) \leq \log^2 k$ and $m \mapsto R^*(k,m)$ is decreasing.
%Therefore
%\[
%\tilde{R}^*(k,2n)=\sum_{m > n}R^*(k,m)\poi(n,m)+\sum_{0\le m \le n}R^*(k,m)\poi(n,m)\le R^*(k,n)+\prob{\Poi(2n)\le n}\log^2k.
%\]
%Then Chernoff bound (see, \eg, \cite[Theorem 5.4]{MU06}) yields  $ \prob{\Poi(2n)\le n}\le \exp(-(1-\log 2)n) $, which implies the left inequality of \prettyref{eq:RRt}.
%
%The right inequality of \prettyref{eq:RRt} follows from Markov's inequality:
%\[
%\tilde{R}^*(k,n/2)=\sum_{m > n}R^*(k,m)\poi(n/2,m)+\sum_{0\le m \le n}R^*(k,m)\poi(n/2,m)\ge R^*(k,n)/2,
%\]

Here we prove the inequality \prettyref{eq:RRt} relating the minimax risk of the entropy estimation under the usual \iid sampling model \prettyref{eq:Rkn} to that under the Poisson sampling model \prettyref{eq:Rknt}. To this end, it is convenient to express the estimator as a function of the original samples instead of the sufficient statistic (histogram). Let $n' \sim \Poi(n)$ and $\{X_1,\ldots\}$ be an \iid\ sequence drawn from $P$ independently of $n'$. Then
\begin{align*}
R^*(k,n) = & ~ \inf_{\hat{H}_n}\sup_{P \in \calM_k}\Expect[( \hat{H}_n(X_1,\ldots,X_n)-H(P) )^2]  \\
\tR^*(k,n) = & ~ \inf_{\{\hat{H}_m\}}\sup_{P \in \calM_k}\Expect[( \hat{H}_{n'}(X_1,\ldots,X_{n'})-H(P) )^2]  
\end{align*}
where  $\hat{H}_m: [k]^m \to \reals_+$. Recall that $0 \leq R^*(k,m) \leq \log^2 k$ and $m \mapsto R^*(k,m)$ is decreasing.
Therefore
\[
\tilde{R}^*(k,2n) \leq \sum_{m > n}R^*(k,m)\poi(2n,m)+\sum_{0\le m \le n}R^*(k,m)\poi(2n,m)\le R^*(k,n)+\prob{\Poi(2n)\le n}\log^2k.
\]
Then Chernoff bound (see, \eg, \cite[Theorem 5.4]{MU06}) yields  $ \prob{\Poi(2n)\le n}\le \exp(-(1-\log 2)n) $, which implies the left inequality of \prettyref{eq:RRt}.

The right inequality of \prettyref{eq:RRt} is slightly more involved. First, by the minimax theorem (cf.~\eg \cite[Theorem 46.5]{Strasser85}), 
\begin{equation}
        R^*(k,n) = \sup_{\pi} \inf_{\hat{H}_n} \Expect[( \hat{H}_n(X_1,\ldots,X_n)-H(P) )^2]  \\
        \label{eq:Rkn-minimax}
\end{equation}
where $\pi$ ranges over all probability distributions (priors) on the simplex $\calM_k$ and the expectation is over $P \sim \pi$ and $X_1,\ldots \iiddistr P$ conditioned on $P$.

Fix a prior $\pi$ and an arbitrary sequence of estimators $\{\hat{H}_m\}$ indexed by the sample size $m$. It is a priori unclear whether the sequence of Bayes risks $\alpha_m \triangleq \Expect[( \hat{H}_m(X_1,\ldots,X_m)-H(P) )^2]$ need be decreasing in $m$. Nevertheless, we can define another sequence of estimators $\{\tilde H_m\}$ which enjoy the desired monotonicity. 
Define $\{\tilde \alpha_m\}$ by $\tilde \alpha_0=\alpha_0$ and $\tilde \alpha_m \triangleq \min_{i\in[m]} \alpha_i = \tilde \alpha_{m-1} \wedge \alpha_m$.
Iteratively define
\[
\tilde H_{m}(x_1,\ldots,x_m) \triangleq
\begin{cases}
\tilde H_{m-1}(x_1,\ldots,x_{m-1}) & \alpha_m \geq \tilde \alpha_{m-1} \\
H_{m}(x_1,\ldots,x_m) & \alpha_m < \tilde \alpha_{m-1} 
\end{cases}, \quad x_1,\ldots,x_m \in [k],
\]
whose Bayes risk is no worse than that of $\hat H_m$.
Then for $n'\sim \Poi(n/2)$ and $P\sim \pi$, 
\begin{align*}
 &~\Expect[( \hat{H}_{n'}(X_1,\ldots,X_{n'})-H(P) )^2]  \\ 
= & ~ \sum_{m \geq 0} \Expect[( \hat{H}_{m}(X_1,\ldots,X_{n'})-H(P) )^2] \poi(n/2,m)    \geq \sum_{m \geq 0} \Expect[( \tilde{H}_{m}(X_1,\ldots,X_{m})-H(P) )^2] \poi(n/2,m)    \\
\geq & ~ \sum_{m \geq 0}^n \Expect[( \tilde{H}_{m}(X_1,\ldots,X_{m})-H(P) )^2] \poi(n/2,m) \geq \frac{1}{2} \Expect[( \tilde{H}_{n}(X_1,\ldots,X_{n})-H(P) )^2] \\
\geq &~\frac{1}{2} \inf_{\tilde{H}_{n}} \Expect[( \tilde{H}_{n}(X_1,\ldots,X_{n})-H(P) )^2],
\end{align*}
where we have used Markov's inequality to conclude $\prob{\Poi(n/2) \geq n} \leq \frac{1}{2}$.
Infimizing the left-hand side over $\{\hat{H}_m\}$, we have
\begin{align}
 \inf_{\{\hat{H}_m\}} \Expect[( \hat{H}_{n'}(X_1,\ldots,X_{n'})-H(P) )^2] \geq \frac{1}{2} \inf_{\tilde{H}_{n}} \Expect[( \tilde{H}_{n}(X_1,\ldots,X_{n})-H(P) )^2].
 \label{eq:ttr}
\end{align}
In view of \prettyref{eq:Rkn-minimax}, supremizing both sides of \prettyref{eq:ttr} over $\pi$ and using the Bayes risk as a lower found for the minimax risk, we conclude that
\begin{align*}
\tilde{R}^*(k,n/2) \geq R^*(k,n)/2.
\end{align*}

\section{Proof of the lower bound}
\label{sec:pf-lb}
We present the proof of the minimax lower bound in \prettyref{sec:pf-lb1} and \prettyref{sec:pf-lb2}; proofs of all auxiliary lemmas are given in \prettyref{sec:pf-lb3}.

\subsection{Proof of \prettyref{prop:lb1}}
\label{sec:pf-lb1}
\begin{proof}
    For any pair of distributions $P$ and $Q$, 
    Le Cam's two-point method (see, \eg, \cite[Section 2.4.2]{Tsybakov09}) yields
    \begin{equation}
        R^*(k,n) \geq \frac{1}{4} (H(P)-H(Q))^2 \exp(-nD(P\|Q)).
        \label{eq:lb-KL}
    \end{equation}
    Therefore it boils down to solving the optimization problem: 
    \begin{equation}
    \sup\{H(P)-H(Q): D(P\|Q) \leq 1/n\}.
    \label{eq:HD}
\end{equation}
    Without loss of generality, assume that $k\geq 2$. 
    % Let 
    % $P=(\frac{1}{2(k-1)},\ldots,\frac{1}{2(k-1)},\frac{1}{2})$ and
    % $Q=(\frac{1+\epsilon}{2(k-1)},\ldots,\frac{1+\epsilon}{2(k-1)},\frac{1-\epsilon}{2})$, where $\epsilon \in (0,\frac{1}{2})$ is to be specified. Direct computation yields
    % $D(P\|Q) = \frac{1}{2} \log \frac{1}{1-\epsilon^2} \leq 2 \epsilon^2$ 
    % and 
    % $|H(P)-H(Q)| = \epsilon \log(k-1) - (1-\epsilon) \log(1-\epsilon)-(1+\epsilon) \log(1+\epsilon) \geq \epsilon \log(k-1) - 2 \epsilon^2$.
    Fix an $\epsilon \in (0,1)$ to be specified. 
    Let 
    \begin{equation}
    P=\pth{\frac{1}{3(k-1)},\ldots,\frac{1}{3(k-1)},\frac{2}{3}}, \quad
    Q=\pth{\frac{1+\epsilon}{3(k-1)},\ldots,\frac{1+\epsilon}{3(k-1)},\frac{2-\epsilon}{3}}.    
    \label{eq:PQ}
\end{equation}
    Direct computation yields
    $D(P\|Q) = \frac{2}{3} \log \frac{2}{2-\epsilon} + \frac{1}{3} \log \frac{1}{\epsilon+1}  \leq \epsilon^2$
    and 
    $H(Q)-H(P) = \frac{1}{3} (\epsilon \log(k-1) + \log 4 + (2-\epsilon) \log \frac{1}{2-\epsilon} + (1+\epsilon) \log \frac{1}{\epsilon+1})
    \geq \frac{1}{3} \log(2(k-1)) \epsilon- \epsilon^2$.
    Choosing $\epsilon = \frac{1}{\sqrt{n}}$ and applying \prettyref{eq:lb-KL}, we obtain the desired \prettyref{eq:lb1}.
\end{proof} 

\begin{remark}
In view of the Pinsker inequality $D(P\|Q) \geq 2 \TV^2(P,Q)$ \cite[p. 58]{ckbook} as well as the continuity property of entropy with respect to the total variation distance: $|H(P)-H(Q)| \leq \TV(P,Q) \log \frac{k}{\TV(P,Q)}$ for $\TV(P,Q) \leq \frac{1}{4}$ 
    \cite[Lemma 2.7]{ckbook}, we conclude that the best lower bound given by the two-point method, \ie, the supremum in \prettyref{eq:HD}, is on the order of $\frac{\log k}{\sqrt{n}}$. Therefore the choice of the pair \prettyref{eq:PQ} is optimal.
%    Interestingly, the Cram\'er-Rao inequality for unbiased estimators gives the same lower bound as \prettyref{eq:lb1}, which, of course, does not directly constitute a minimax lower bound.

    \label{rmk:HD}
\end{remark}

\subsection{Proof of \prettyref{prop:lb2}}
\label{sec:pf-lb2}
For $0<\epsilon<1$, define the set of \emph{approximate} probability vectors by
\begin{equation}
    \calM_k(\epsilon) \triangleq \sth{P\in \reals_+^k: \abs{\sum_{i=1}^k p_i-1}\le \epsilon }.    
    \label{eq:Mkeps}
\end{equation}
which reduces to the probability simplex $\calM_k$ if $\epsilon=0$.

Generalizing the minimax quadratic risk \prettyref{eq:Rknt} for Poisson sampling, we define
\begin{equation}
    \tilde{R}^*(k,n,\epsilon) \triangleq \inf_{\hat{H}'}\sup_{P \in \calM_k(\epsilon)}\Expect( \hat{H}'(N)-H(P) )^2,    
    \label{eq:Rkn1}
\end{equation}
where $ N=(N_1,\dots,N_k) $ and $ N_i\inddistr \Poi(np_i) $ for $ i=1, \dots, k $.
%Here we loosen the constraint on input distribution for the \iid ensemble in \prettyref{sec:iid},
Since $ P $ is not necessarily normalized, $ H(P) $ may not carry the meaning of entropy. Nevertheless, $ H $ is still valid  a functional.
The risk defined above is connected to the risk \prettyref{eq:Rkn} for multinomial sampling by the following lemma:
\begin{lemma}
    \label{lmm:poisson}
    For any $k,n \in \naturals$ and $ \epsilon< 1/3$, 
    \[
    R^*(k,n/2) \ge \frac{1}{3} \tilde{R}^*(k,n,\epsilon) -  (\log k )^2\exp(-n/50) -  \pth {\epsilon\log k}^2 
    -  \pth {(1+\epsilon)\log (1+\epsilon)}^2.
    \]
\end{lemma}

To establish a lower bound of $ \tilde{R}^*(k,n,\epsilon) $, we apply generalized Le Cam's method involving two composite hypothesis as in \prettyref{eq:compHT}, which entails choosing two priors such that the entropy values are separated with probability one. It turns out that this can be relaxed to separation \emph{on average}, if we can show that the entropy values are concentrated at their respective means.
%. Indeed, given two priors, if we can show that the functional values are concentrated at their respective means, 
% due to law of large numbers. 
%we then obtain a lower bound given by the separation in the \emph{mean} functional values. 
This step is made precise in the next lemma:
\begin{lemma}
    \label{lmm:Rknet}
    Let $ U $ and $ U' $ be random variables such that $ U,U'\in[0,\lambda] $ and $ \expect{U}=\expect{U'}\le 1 $
    and $ \abs{\expect{\phi(U)}-\expect{\phi(U')}}\ge d $, where $ \lambda< k/e$. Let $ \epsilon=\frac{4\lambda}{\sqrt{k}} $.
    Then 
    \begin{equation}
        \tilde{R}^*(k,n,\epsilon)
        \ge \frac{d^2}{16}\pth{\frac{7}{8}-k\TV(\expect{\Poi\pth{nU/k}}, \expect{\Poi\pth{nU'/k}})
            - \frac{32\lambda^2\log^2\frac{k}{\lambda}}{kd^2} }.
        \label{eq:Rknet}
    \end{equation}
\end{lemma}
% The following result makes rigorous the \iid\ construction of priors described in \prettyref{sec:prior}.
% \begin{theorem}
%     % Let $ d\asymp \frac{k}{n\log k} $ and $ \lambda\asymp \frac{k\log k}{n} $.
%     Let $ U $ and $ U' $ be random variables such that $ U,U'\in[0,\lambda] $ and $ \expect{U}=\expect{U'}\le 1 $
%     and $ \abs{\expect{\phi(U)}-\expect{\phi(U')}}\ge d $, where $ \lambda< c'k $ for some small constant $ c' $.
%     Then 
%     \begin{equation}
%         \begin{aligned}
%             {R}^*(k,n/2)
%             \geq
%             &\frac{d^2}{96}\pth{\frac{7}{8}-k\TV(\expect{\Poi\pth{nU/k}}, \expect{\Poi\pth{nU'/k}})
%                 - \frac{32\lambda^2\log^2\frac{k}{\lambda}}{kd^2} }\\
%             &-\log^2k\exp\pth{-\frac{n}{50}}-\frac{16\lambda^2 \log^2k}{k}-\pth{1+\frac{\lambda}{\sqrt {k}}}\log^2\pth{1+\frac{\lambda}{\sqrt {k}}}.
%         \end{aligned}
%         \label{eq:iid}
%     \end{equation}
%     \label{thm:iid}
% \end{theorem}

The following result gives a sufficient condition for Poisson mixtures to be indistinguishable in terms of moment matching.
Analogous results for Gaussian mixtures have been obtained 
in \cite[Section 4.3]{LNS99} using Taylor expansion of the KL divergence and orthogonal basis expansion of $\chi^2$-divergence
in \cite[Proof of Theorem 3]{CL11}. For Poisson mixtures we directly deal with the total variation as the $\ell_1$-distance between the mixture probability mass functions.

% \begin{lemma}
%     Let $ \lambda=\frac{c_1k\log k}{n}, L=c_0\log k $, where $ c_1< c_0/2 $.
%     Let $U$ and $U'$ be random variables on $ \qth{0,\lambda} $. 
%     If $ \expect{U^j} = \expect{U'^j}, ~ j = 1,\ldots,L $, 
%     then
%     \begin{equation}
%         \TV(\expect{\Poi\pth{nU/k}}, \expect{\Poi\pth{nU'/k}})
%         \le 2\exp\pth{-\log k\pth{\frac{c_0}{2}\log\frac{c_0}{2ec_1}-c_1}},
%         \label{eq:tv-bound}
%     \end{equation}
%     \label{lmm:tv-bound}
% \end{lemma}

\begin{lemma}
    \label{lmm:tv-bound}
    Let $ V $ and $ V' $ be random variables on $ [0,M] $.
    If $ \Expect[V^j]=\Expect[V'^j],~j=1,\dots,L $ and $ L> 2e M $,
    then
    \begin{equation}
        \TV(\Expect[\Poi(V)],\Expect[\Poi(V')]) \le \pth{\frac{2eM}{L}}^L.
        % \TV(\expect{\Poi(V)},\expect{\Poi(V')}) \le 2\exp\pth{-\pth{\frac{L}{2}\log\frac{L}{2e\lambda}-\lambda}}\wedge 1.
        \label{eq:tv-bound}
    \end{equation}
\end{lemma}

\begin{remark}
        \label{rmk:tv-bound}
        In an earlier version of the paper,\footnote{See Lemma 3 in \url{http://arxiv.org/pdf/1407.0381v2.pdf}.} the following weaker total variation bound
       \begin{equation}
        \TV(\expect{\Poi(V)},\expect{\Poi(V')}) \le 2\exp\pth{-\pth{\frac{L}{2}\log\frac{L}{2eM}-M}}\wedge 1,
       \label{eq:tv-bound-old}
\end{equation}
        was proved by truncating the summation in the total variation.
        This bound suffices for our purpose; in fact, the same proof techniques have been subsequently used in \cite[Lemma 11]{JVHW15} for minimax lower bound of estimating other functionals. Nevertheless, \prettyref{eq:tv-bound} provides a strict improvement over \prettyref{eq:tv-bound-old}, whose proof is even simpler and involves no truncation argument.
        What remains open is the optimal number of matching moments to ensure indistinguishability of the Poisson mixtures. \prettyref{lmm:tv-bound} implies that as soon as $L/M$ exceeds $2e$ the total variation decays exponentially; it is unclear whether $L$ needs to grow linearly with $M$ in order to drive the total variation to zero.
\end{remark}

To apply \prettyref{lmm:Rknet} and \prettyref{lmm:tv-bound} we need to construct two random variables, namely $ U $ and $ U' $, that have matching moments of order $ 1,\dots,L $,
and large discrepancy in the mean functional value $ \abs{\expect{\phi(U)}-\expect{\phi(U')}}$, as described in \prettyref{sec:prior} and formulated in \prettyref{eq:FL}.
    As shown in \prettyref{app:moments}, we can obtain $U,U'$ with matching moments from the dual of the best polynomial approximation of $ \phi $, namely \prettyref{eq:EL}; however, we have little control over the value of the common mean $\Expect[U]=\Expect[U']$ and it is unclear whether it is less than one as required by \prettyref{lmm:tv-bound}.
Of course we can normalize $U,U'$ by their common mean which preserves moments matching; however, the mean value separation $ \abs{\expect{\phi(U)}-\expect{\phi(U')}}$ also shrinks by the same factor, which results in a suboptimal lower bound. 
    
    To circumvent this issue, we first consider 
    auxiliary random variables $X,X'$ supported on a interval bounded away from $ 0 $; 
    leveraging the property that their ``zeroth moments'' are one, we then construct the desired random variables $U,U'$ via a change of measure.
    To be precise, given $\eta \in (0,1)$ and any random variables $ X,X'\in[\eta,1] $ that have matching moments up to the $ L\Th $ order, we can construct $U,U'$ from $X,X'$ with the following distributions
    \begin{equation}
        \begin{aligned}
            & P_U(\diff u)
            =\pth{1-\expect{ \frac{\eta}{X} } }\delta_0(\diff u)+\frac{\alpha}{u}P_{\alpha X/\eta}(\diff u),\\
            & P_{U'}(\diff u)
            =\pth{1-\expect{ \frac{\eta}{X'} } }\delta_0(\diff u)+\frac{\alpha}{u}P_{\alpha X'/\eta}(\diff u),
            % P_U(du)= 0.5(1-E[\eta/X])\delta_0(du)+\alpha/2u P_{\alpha X/\eta}(du)+0.5\delta_{2-\alpha}(du)
        \end{aligned}
        \label{eq:UU}
    \end{equation}
    for some fixed $ \alpha\in(0,1) $.
    Since $ X,X'\in \qth{\eta,1} $ and thus $ \expect{\frac{\eta}{X}}, \expect{\frac{\eta}{X'}}\le 1 $, these distributions are well-defined and supported on $\qth{0,\alpha \eta^{-1}} $. Furthermore, 
    \begin{lemma}
    % $ \expect{\phi(V)}=\log\eta+\expect{\log \frac{1}{X}}, \expect{\phi(V')}=\log\eta+\expect{\log \frac{1}{X'}}$ 
    $ \expect{\phi(U)} - \expect{\phi(U')} = \alpha(\Expect[\log\frac{1}{X}]-\Expect[\log\frac{1}{X'}])$ 
    and 
    $ \expect{U^j} = \expect{U'^j}, ~ j = 1,\ldots,L+1. $ In particular, $ \expect{U}=\expect{U'}=\alpha $. 
    \label{lmm:UUXY}
\end{lemma}
%    Now $ U,U' $ still have match moments and have common mean $ \alpha $.
%    Note that $ \Expect[\phi(U)]-\Expect[\phi(U')]=\alpha(\Expect[\log\frac{1}{X}]-\Expect[\log\frac{1}{X'}]) $, 
To choose the best $X,X'$, we consider the following auxiliary optimization problem over random variables $X$ and $X'$ (or equivalently, the distributions thereof).
\begin{equation}
    \begin{aligned}
        \calE^* = \max & ~ \expect{\log \frac{1}{X}} - \expect{\log \frac{1}{X'}}  \\
        \text{s.t.}     & ~ \expect{X^j} = \expect{X'^j}, \quad j = 1,\ldots,L, \\
        & ~ X,X' \in [\eta,1],
    \end{aligned}
    \label{eq:Rstar}
\end{equation}
where $ 0<\eta<1 $. 
Note that \prettyref{eq:Rstar} is an infinite-dimensional linear programming problem with finitely many constraints. Therefore it is natural to turn to its dual. In \prettyref{app:moments} we show that the maximum $\calE^*$ exists and coincides with twice the best $L_\infty$ approximation error of the $\log$ over the interval $[\eta,1]$ by polynomials of degree $L$:
\begin{equation}
        \calE^*=2E_{L}(\log, [\eta,1]).
        \label{eq:RE}
    \end{equation}
  By definition, this approximation error 
%  on the right-hand side of \prettyref{eq:RE} 
  is decreasing in the degree $L$ when $\eta$ is fixed; on the other hand, since the logarithm function blows up near zero, for fixed degree $L$ the approximation error also diverges as $\eta$ vanishes. As shown in \prettyref{app:error}, in order for the error to be bounded away from zero which is needed in the lower bound, it turns out  that the necessary and sufficient condition 
   is when $\eta$ decays according to $L^{-2}$:
% We need the following lemma to prove our lower bound:
\begin{lemma}%[Separation in functional value]
    \label{lmm:sep}
    There exist universal positive constants $c , c', L_0$ such that for any $L\geq L_0$,
    \begin{equation}
        E_{\floor{c L}}(\log, [L^{-2},1]) \geq c'.
        \label{eq:sep}
    \end{equation}
\end{lemma}

\begin{proof}[Proof of \prettyref{prop:lb2}]
% In order to apply the iid construction in \prettyref{thm:iid}, the random variable needs to have unit mean. However, we have no control over the means of the maximizer of \prettyref{eq:Rstar}. 
% Fix $\alpha \in (0,1)$ to be specified later.
 Let $X$ and $X'$ be the maximizer of \prettyref{eq:Rstar}.
 Now we construct $U$ and $U'$ from $X$ and $X'$ according to the recipe \prettyref{eq:UU}.
%  with the following distributions
% \begin{equation}
%     \begin{aligned}
%         & P_U(\diff u)
%         =\pth{1-\expect{ \frac{\eta}{X} } }\delta_0(\diff u)+\frac{\alpha}{u}P_{\alpha X/\eta}(\diff u),\\
%         & P_{U'}(\diff u)
%         =\pth{1-\expect{ \frac{\eta}{X'} } }\delta_0(\diff u)+\frac{\alpha}{u}P_{\alpha X'/\eta}(\diff u).
%         % P_U(du)= 0.5(1-E[\eta/X])\delta_0(du)+\alpha/2u P_{\alpha X/\eta}(du)+0.5\delta_{2-\alpha}(du)
%     \end{aligned}
%     \label{eq:UU}
% \end{equation}
% Since $ X,X'\in \qth{\eta,1} $ and thus $ \expect{\frac{\eta}{X}}, \expect{\frac{\eta}{X'}}\le 1 $, these distributions are well-defined and $ U,U'\in \qth{0,\alpha \eta^{-1}} $.
By \prettyref{lmm:UUXY}, the first $L+1$ moments of $ U $ and $ U' $ are matched with means equal to $ \alpha $ which is less than one; moreover,
\begin{equation}        
\expect{\phi(U)} - \expect{\phi(U')} = \alpha \calE^*.
        \label{eq:UUsep}
\end{equation} 
% Finally, we set
% \begin{align}
%     P_U = \frac{1}{2} (P_{V} + \delta_{2-\alpha})
%     , \quad P_{U'} = \frac{1}{2} (P_{V'} + \delta_{2-\alpha} ).
% \end{align}

Recall the universal constants $ c $ and $c' $ defined in \prettyref{lmm:sep}. If $n \geq \frac{2k }{\log k}$,
let $ c_1 \leq 2$ be a constant satisfying 
$ \frac{c}{2}\log\frac{c}{4ec_1}>2 $ and thus $ c>4ec_1 $. % $ \frac{c}{4}\log\frac{c}{4ec_1}-c_1>2 $ and $ c_1< \frac{c}{4e} $.
Let $ \eta=\log^{-2}k $, $ L=\floor{c \log k}\le\frac{c\log k}{2}$, $ \alpha=\frac{c_1k}{n\log k} $ and $ \lambda=\alpha\eta^{-1}=\frac{c_1k\log k}{n} $. Therefore $\alpha \leq 1$.
Using \prettyref{eq:UU} and \prettyref{eq:UUsep}, we can construct two random variables $ U,U' \in[0,\lambda]  $ such that $\Expect[U]=\Expect[U']=\alpha$, $\Expect[U^j]=\Expect[U'^j]$, for all  $j \in [L]$, and $ \expect{\phi(U)} - \expect{\phi(U')} = \alpha \calE^*$. It follows from \prettyref{eq:RE} and \prettyref{lmm:sep} that $ \calE^*\geq 2c'$ and thus $ \abs{\expect{\phi(U)} - \expect{\phi(U')}}\ge 2c'\alpha $.
By the choice of $ c_1 $, applying \prettyref{lmm:tv-bound} yields $ \TV(\expect{\Poi\pth{nU/k}}, \expect{\Poi\pth{nU'/k}})\le 2k^{-2} $. Finally, applying \prettyref{lmm:poisson} and \prettyref{lmm:Rknet} with $ d=2c'\alpha $ yields the desired lower bound $ R^*(k,n/2)\gtrsim \alpha^2\asymp(\frac{k}{n \log k})^2 $. Consequently, $R^*(k,n)\gtrsim (\frac{k}{n \log k})^2 $ when $ n \ge \frac{k }{\log k} $.
If $ n\le \frac{k}{\log k}$ by monotonicity, $ R^*(k,n) \ge R^*(k,\frac{k}{\log k}) \gtrsim 1 $.
\end{proof}

\begin{remark}[Structure of the least favorable priors]
    From the proof of \prettyref{eq:RE} in \prettyref{app:moments}, we conclude that $ X,X'$ are in fact discrete random variables with disjoint support each of which has $L+2 \asymp \log k$ atoms. Therefore $ U,U' $ are also finitely-valued; however, our proof does not rely on this fact. Nevertheless,  it is instructive to discuss the structure of the prior. 
%: $ \expect{H(\sfP)}=\expect{\phi(U)}+\expect{U}\log k+\phi(1-\alpha) $. 
Except for possibly a fixed large mass, the masses of random distributions $\sfP$ and $\sfP'$ are drawn from the distribution $U$ and $U'$ respectively, which 
lie in the interval $[0,\frac{\log k}{n}]$. Therefore, although $\sfP$ and $\sfP'$ are distributions over $k$ elements, they only have $\log k$ distinct masses and the locations are randomly permuted. 
Moreover, the entropy of $\sfP$ and $\sfP'$ constructed based on $U$ and $U'$ (see \prettyref{eq:PP}) are concentrated near the respective mean values, both of which are close to $\log k$ but differ by a constant factor of $\frac{k}{n \log k}$.

   \label{rmk:prior}
\end{remark}

\subsection{Proof of Lemmas}
\label{sec:pf-lb3}
\begin{proof}[Proof of \prettyref{lmm:poisson}]
    Fix $\delta>0$.
    Let $\hat{H}(\cdot,n)$ be a near-minimax entropy estimator for fixed sample size $n$, \ie,
    \begin{equation}
        \sup_{P \in \calM_k} \Expect[(\hat{H}(N,n)-H(P))^2] \leq \delta + R^*(k,n).    
        \label{eq:delta-minimax}
    \end{equation}
    Using these estimators we construct a estimator for the Poisson model in \prettyref{eq:Rknet}. 
    Fix an arbitrary $P=(p_1,\ldots,p_k)\in \calM_k(\epsilon)$. Let $N=(N_1,\ldots,N_k)$ with $N_i \inddistr \Poi(n p_i)$ and let $n' = \sum N_i$.
    We construct an estimator for the Poisson sampling model by 
    % and minimax risk $ \tilde{R}^*(k,n,\epsilon) $
    % with minimax entropy estimator $ \hat{H}^ $ for $ R\pth {k,\sum_{i=1}^{k}N_i} $ by
    \[
    \tilde{H}(N)= \hat{H}(N,n').
    \]
    % \nb{existence of $ \hat{H}^ $?} \nbr{I am not sure. Let's circumvent this question by taking the inf?}
    The functional $ H $ is related to entropy of the normalized $ P $ by
    \begin{equation}
        H(P)=\sum_{i=1}^{k}p_i\log \frac{1}{p_i}= \pth{\sum_ip_i} \log \frac{1}{ \sum_ip_i } + \pth{\sum_ip_i}H\pth{ \frac{P}{ \sum_ip_i } }.
        \label{eq:H-unnormalized}
    \end{equation}
    Then triangle inequality and \prettyref{eq:H-unnormalized} give us
    \begin{align}
        & ~ \frac{1}{3}( \tilde{H}(N) - H(P) )^2\nonumber\\
        \le & ~  \pth{ \tilde{H}(N)-H\pth{ \frac{P}{\sum_ip_i} } }^2 
        +  \pth{ \pth{1-\sum_ip_i} H\pth{ \frac{P}{ \sum_ip_i } } }^2 
        +  \pth{ \pth{\sum_ip_i} \log \frac{1}{ \sum_ip_i }  }^2\nonumber\\
        \le & ~  \pth{ \tilde{H}(N)-H\pth{ \frac{P}{\sum_ip_i} } }^2 
        +  \pth {\epsilon\log k}^2 
        +  \pth {(1+\epsilon)\log (1+\epsilon)}^2.\label{eq:H-triangle}
    \end{align}
    % for all $ \epsilon<1 $. 
    For the first term of \prettyref{eq:H-triangle}, we observe that
    conditioned on $ n'=m $, $ N\sim \Multinom\pth{ m, \frac{P}{\sum_i p_i} } $.
    % Also note that $ R(k,m)=\sup_{p\in \reals_+^k: \sum_i^kp_i=1}\expect{ \hat{H}(N)-H(P) }^2 $ for $ N\sim \Multinom\pth{ m, p } $,
    Therefore in view of the performance guarantee in \prettyref{eq:delta-minimax}, we obtain that
    \begin{align*}
        \Expect\pth{ \tilde{H}(N)-H\pth{ \frac{P}{\sum_ip_i} } }^2 
        = & \sum_{m=0}^{\infty} \Expect \qth{ \pth{ \hat{H}(N,m)-H\pth{ \frac{P}{\sum_ip_i} } }^2 \Bigg\vert n'=m } \prob{n'=m} \\ 
        \le & \sum_{m=0}^{\infty} R^*(k,m) \prob{n'=m} + \delta. 
    \end{align*}
    Now note that for fixed $ k $, the minimax risk $n \mapsto R^*(k,n) $ is decreasing  and $0 \leq R^*(k,n)\le \pth { \log k }^2 $.
    Since $n' = \sum_{i=1}^{k}N_i\sim \Poi\pth{ n\sum_ip_i } $ and $ \abs{\sum_i^kp_i-1}\le \epsilon \leq 1/3$, we have
    \begin{align}
        \Expect\pth{ \hat{H}(N)-H\pth{ \frac{P}{\sum_ip_i} } }^2 
        \le & \sum_{m\ge n/2} R^*(k,m)\prob{n'=m} + \pth { \log k }^2 \prob{n'\le \frac{n}{2}}+\delta\nonumber\\
        \le & R^*(k,n/2) + (\log k )^2\exp(-n/50)+\delta,\label{eq:H-multinom}
    \end{align}
    where in the last inequality we used the Chernoff bound (see, \eg, \cite[Theorem 5.4]{MU06}).
    Plugging \prettyref{eq:H-multinom} into \prettyref{eq:H-triangle} and by the arbitrariness of $ \delta $, the lemma follows.
    % yields $ \prob{\sum_{i=1}^{k}N_i\le \frac{n}{2}}<\exp(-n/50) $. 
\end{proof}

\begin{proof}[Proof of \prettyref{lmm:Rknet}]
    Denote the common mean by $ \alpha\triangleq\expect{U}=\expect{U'}\le 1 $.
    Define two random vectors 
    \begin{equation}
      \sfP=\pth{ \frac{U_1}{k},\dots,\frac{U_k}{k},1-\alpha } , \quad \sfP'=\pth{ \frac{U_1'}{k},\dots,\frac{U_k'}{k},1-\alpha },
    \label{eq:PP}
\end{equation}
    where $ U_i, U_i'$ are \iid copies of $ U, U'$, respectively. Note that $ \epsilon= \frac{4\lambda}{\sqrt {k}}\ge 4\sqrt {\frac{\var [U] \vee \var [U']}{k}} $.
    Define the following events indicating that $ U_i $ and $ H(\sfP) $ are concentrated near their respective mean values:
    $$ E\triangleq\sth{\abs{\sum_i\frac{U_i}{k}-\alpha}\le \epsilon, \abs{H(\sfP)-\expect{H(\sfP)}}\le\frac{d}{4} } ,
    E'\triangleq\sth{\abs{\sum_i\frac{U_i'}{k}-\alpha}\le \epsilon, \abs{H(\sfP')-\expect{H(\sfP')}}\le\frac{d}{4} } .$$
    Using the independence of $ U_i $, Chebyshev's inequality and union bound yield that
    \begin{align}
        \prob{E^c}
        \le&~ \prob{\abs{\sum_i\frac{U_i}{k}-\alpha}> \epsilon}+\prob{\abs{H(\sfP)-\expect{H(\sfP)}}>\frac{d}{4}}\nonumber\\
        \le&~ \frac{\var[U]}{k\epsilon^2}+\frac{16\sum_i\var[\phi(U_i/k)]}{d^2}
        \le \frac{1}{16}+\frac{16\lambda^2\log^2\frac{k}{\lambda}}{kd^2},\label{eq:Ec}
    \end{align}
    where the last inequality follows from the fact that
    $ \var\qth{\phi\pth{\frac{U_i}{k}}}\le \expect{\phi\pth{\frac{U_i}{k}}}^2\le \pth{\phi\pth{\frac{\lambda}{k}}}^2 $ when $\lambda/k<e^{-1} $ by assumption. By the same reasoning,
    \begin{equation}
        \prob{E'^c}
        \le \frac{1}{16}+\frac{16\lambda^2\log^2\frac{k}{\lambda}}{kd^2}.
        \label{eq:Epc}
    \end{equation}
    Note that conditioning on $ E $ and $ E' $ the random vectors in \prettyref{eq:PP} belong to $\calM_k(\epsilon)$.
    Now we define two priors on the set $\calM_k(\epsilon)$ using \prettyref{eq:PP} with the following conditional distributions:
    \[
    \pi = P_{\sfP|E},\quad\pi'=P_{\sfP'|E'}.
    \]

    % First we consider the separation of the functional values under $ \pi,\pi' $.
    It follows from $ H(\sfP)=\frac{1}{k}\sum_i\phi(U_i)+\frac{\log k}{k}\sum_iU_i+\phi(1-\alpha) $ that
    $ \expect{H(\sfP)}=\expect{\phi(U)}+\expect{U}\log k+\phi(1-\alpha) $. Similarly,
    $ \expect{H(\sfP')}=\expect{\phi(U')}+\expect{U'}\log k+\phi(1-\alpha) $.
    By assumption $ |\expect{H(\sfP)}-\expect{H(\sfP')}|=|\expect{\phi(U)}-\expect{\phi(U')}|\ge d $.
    By the definition of events $ E,E' $ and triangle inequality, we obtain that under $ \pi,\pi' $
    \begin{equation}
        \abs{H(\sfP)-H(\sfP')}\ge \frac{d}{2}.
        \label{eq:H-sep}
    \end{equation}

    Now we consider the total variation of the sufficient statistics $N=(N_i)$ under two priors. Note that conditioned on $p_i$, we have $ N_i\sim \Poi(np_i) $.
    The triangle inequality of total variation then yields
    \begin{align}
        \TV \pth{ P_{N|E}, P_{N'|E'} }
        \le & \TV \pth{ P_{N|E}, P_N } + \TV \pth{ P_{N}, P_{N'} } + \TV \pth{  P_{N'} , P_{N'|E'}}\nonumber\\
        = & \Prob\qth{ E ^c} + \TV \pth{ P_{N}, P_{N'} } + \Prob\qth{ E'^c}\nonumber\\
        \le &  \TV \pth{ P_{N}, P_{N'} } + \frac{1}{8}+\frac{32\lambda^2\log^2\frac{k}{\lambda}}{kd^2},\label{eq:tvne}         
    \end{align}
    where in the last inequality we have applied \prettyref{eq:Ec}--\prettyref{eq:Epc}.
    Note that $ P_{N},P_{N'} $ are marginal distributions under priors $ P_{\sfP},P_{\sfP'} $ respectively.
    In view of the fact that the total variation between product distributions is at most the sum of total variations of pair of marginals, we obtain 
    \begin{align}
        \TV \pth{ P_{N}, P_{N'} }
        \le& \sum_{i=1}^{k}\TV \pth{ P_{N_i}, P_{N_i'} }+\TV(\Poi(n(1-\alpha)),\Poi(n(1-\alpha)))\nonumber\\
        =& k\TV(\expect{\Poi\pth{nU/k}}, \expect{\Poi\pth{nU'/k}}).\label{eq:tvnn}
    \end{align}
Then it follows from \prettyref{eq:H-sep}--\prettyref{eq:tvnn} and Le Cam's lemma \cite{Lecam86} that
    \begin{equation}
        \tilde{R}^*(k,n,\epsilon)
        \ge \frac{d^2}{16}\pth{\frac{7}{8}-k\TV(\expect{\Poi\pth{nU/k}}, \expect{\Poi\pth{nU'/k}})
            - \frac{32\lambda^2\log^2\frac{k}{\lambda}}{kd^2} }.
        \label{eq:Rnke}
    \end{equation}
\end{proof}

% Hence for all $ \epsilon=o(1/\log k) $, and $ n> 100 \log\log k $
% \[
% \tilde{R}^*(k,n,\epsilon)
% \le \sup_{p\in \reals_+^k: \abs{\sum_i^kp_i-1}< \epsilon }\Expect (\hat{H}(N)-H(P))^2
% \le R(k,n/2) + o(1).
% \]
% In particular, for $ n\asymp k/\log k $, if no consistent estimator exists for estimating $ R'\pth{k,2n,k^{-1/3}} $,
% then it is impossible to find a consistent estimator for $ R(k,n) $.

\begin{proof}[Proof of \prettyref{lmm:tv-bound}]
    By the assumption that $ \Expect [V^j]=\Expect[V'^j] $ when $ j\le L $, we obtain that
    \begin{align*}
      \TV(\Expect[\Poi(V)],\Expect[\Poi(V')])
      =&\frac{1}{2}\sum_{j\ge 0}\abs {\Expect [\poi(V,j)] - \Expect [\poi(V',j)]}\\
      =& \frac{1}{2}\sum_{j \ge 0}\abs{\expect{\sum_{m\ge 0}\frac{(-V)^m}{m!}\frac{V^j}{j!}} - \expect{\sum_{m\ge 0}\frac{(-V')^m}{m!}\frac{V'^j}{j!}} }\\
      =& \frac{1}{2}\sum_{j \ge 0} \frac{1}{j!}\abs{ \sum_{m> L-j}\frac{(-1)^m}{m!} ( \Expect [V ^{m+j}]-\Expect  [V'^{m+j}]) }.
    \end{align*}
    By triangle inequality and the assumption that $ V,V'\in[0,M] $, we have that
    \begin{align*}
      \TV(\Expect[\Poi(V)],\Expect[\Poi(V')])
      \le  \sum_{j \ge 0}\frac{M^j}{j!}\sum_{m> L-j}\frac{M^m}{m!}
      =&e^{2M}\sum_{j \ge 0}\sum_{m> L-j}\Prob[N_1=j,N_2=m]\\
      =&e^{2M}\Prob[N_1+N_2>L],
    \end{align*}
    where $ N_1,N_2\iiddistr \Poi(M) $ and thus $ N_1+N_2\sim \Poi(2M) $.
    Applying Chernoff bound when $ L>2M $ yields that
    \begin{equation*}
        \TV(\Expect[\Poi(V)],\Expect[\Poi(V')])
        \le e^{2M} e^{-2M}\pth{\frac{2eM}{L}}^L
        = \pth{\frac{2eM}{L}}^L.
    \end{equation*}

\end{proof}

\begin{proof}[Proof of \prettyref{lmm:UUXY}]
    Note that
    \begin{equation*}
        \expect{\phi(U)}
        =\int \pth{u\log \frac{1}{u}}\frac{\alpha}{u}P_{\alpha X/\eta}(\diff u)
        =\alpha \expect{\log\frac{\eta}{ \alpha X}}
    \end{equation*}
    and, analogously, $ \expect{\phi(U')}= \alpha \expect{\log\frac{\eta}{\alpha X'}} $.
    Therefore, $ \expect{\phi(U)} - \expect{\phi(U')} = \alpha (\expect{\log\frac{1}{ X}}-\expect{\log\frac{1}{ X'}})$. Moreover, for any $j\in[L+1]$,
    \begin{equation*}
        \expect{U^j} 
        = \int u^j\frac{\alpha}{u}P_{\alpha X/\eta}(\diff u)
        = \expect{(\alpha X/\eta)^{j-1}\alpha}
    \end{equation*}
    which coincides with $ \expect{U'^j}= \Expect \qth{(\alpha X'/\eta)^{j-1}\alpha} $, in view of the moment matching condition of $X$ and $X'$ in \prettyref{eq:Rstar}. 
    In particular, 
    $ \expect{U}=\expect{U'}=\alpha $ follows immediately. 
\end{proof}

\section{Proof of the upper bound}
\label{sec:pf-ub}
    \begin{proof}[Proof of \prettyref{prop:err-rate}]
        Given that $ N_i' $ is above (resp. below) the threshold $ c_2\log k $, we can conclude with high confidence that $ p_i $ is above (resp. below) a constant factor of $ \frac{\log k}{n} $.
        Define two events by $ E_1\triangleq\bigcap_{i=1}^k\sth{N_i' \le c_2\log k \Rightarrow p_i\le \frac{c_1\log k}{n}} $
        and $ E_2\triangleq\bigcap_{i=1}^k\sth{N_i' > c_2\log k \Rightarrow p_i> \frac{c_3\log k}{n}} $, where $ c_1>c_2>c_3 $.
        Applying the union bound and the Chernoff bound for Poissons (\cite[Theorem 5.4]{MU06}) yields that
        \begin{align}
            \prob{E_1^c}
            =& \Prob\qth{\bigcup_{i=1}^k\sth{N_i' \le c_2\log k,  p_i> \frac{c_1\log k}{n}}}\nonumber\\
            \le & k \, \Prob\qth{\Poi(c_1\log k)\le c_2\log k}\nonumber\\
            \le & \frac{1}{k^{c_1-c_2\log\frac{ec_1}{c_2}-1}},
            \label{eq:E1c}
        \end{align}
        and, entirely analogously,
        \begin{equation}
            \Prob[E_2^c]
            \le \frac{1}{k^{c_3+c_2\log\frac{ec_2}{c_3}-1}}.
            \label{eq:E2c}
        \end{equation}
        Define an event $ E\triangleq E_1\cap E_2 $. Again union bound gives us $ \prob{E^c}\le \prob{E_1^c}+\prob{E_2^c} $.

        By construction $ \hat{H}=(\tilde{H}\vee 0) \wedge \log k $, the fact $ H(P)\in[0,\log k] $ yields that $ |H(P)-\hat{H}|\le |H(P)-\tilde{H}| $ and $ |H(P)-\hat{H}|\le \log k $. So the MSE can be decomposed and upper bounded by
        % \begin{equation}
        \begin{align}
            \Expect(H(P)-\hat{H})^2
            =& \Expect[(H(P)-\hat{H})^2\Indc_E]+\Expect[(H(P)-\hat{H})^2\Indc_{E^c}]\nonumber\\
            \le & \Expect[(H(P)-\tilde{H})^2\Indc_E]+(\log k)^2(\prob{E_1^c}+\prob{E_2^c}).
            \label{eq:MSE0}
        \end{align}
        % \end{equation}

        % Now we need to take care of the second order moment of $ (H(P)-\tilde{H})\Indc_E $.

        Define $$ \calE_1\triangleq\sum_{i\in I_1}\phi(p_i)-g_L(N_i) , \quad
        \calE_2\triangleq\sum_{i\in I_2}\pth{\phi(p_i)-\phi\pth{\frac{N_i}{n}}-\frac{1}{2n}} ,$$
        where the (random) index sets defined by
        $$ I_1\triangleq\sth{i:N_i'\le c_2\log k,p_i\le \frac{c_1\log k}{n}} , \quad
         I_2\triangleq\sth{i:N_i'>c_2\log k,p_i> \frac{c_3\log k}{n}} $$
        are independent of $ N $ due to the independence of $N$ and $N'$.
        The implications in the event $ E $ yields 
        \begin{equation}
            (H(P)-\tilde{H})\Indc_E
            =\calE_1\Indc_E+\calE_2\Indc_E.
            \label{eq:E1E2}
        \end{equation}
        Combining \prettyref{eq:MSE0}--\prettyref{eq:E1E2} and applying triangle inequality we obtain that
        \begin{align}
            \Expect(H(P)-\hat{H})^2
            \le 2\Expect[\calE_1^2] + 2\Expect[\calE_2^2] + (\log k)^2(\prob{E_1^c}+\prob{E_2^c}).
            \label{eq:MSE}
        \end{align}
Next we proceed to consider the error terms $ \calE_1 $ and $ \calE_2 $ separately. % and then combine them together in the end.

        \paragraph{Case 1: Polynomial estimator}
        It is known that (see, \eg, \cite[Section 7.5.4]{timan63}) the optimal uniform approximation error of $ \phi $ by degree-$ L $ polynomials on $ [0,1] $ satisfies 
        $ L^2 E_L\pth{\phi,[0,1]} \to c>0$ as $L\to \infty$. Therefore $E_L\pth{\phi,[0,1]} \lesssim L^{-2}.$
        By a change of variables, it is easy to show that 
        $$ E_{L}\pth{\phi,\qth{0,\frac{c_1\log k}{n}}} = \frac{c_1\log k}{n} E_{L}\pth{\phi,\qth{0,1}} \lesssim \frac{1}{n\log k}  .$$ 
        By definition, $ I_1\subseteq \{i:p_i\le \frac{c_1\log k}{n}\} $. 
        Since $ g_L(N_i) $ is an unbiased estimator of $ P_L(p_i) $, the bias can be bounded by the uniform approximation error almost surely as
        \begin{equation}
            |\Expect[\calE_1|I_1]|
            =\abs{\sum_{i\in I_1}p_i\log\frac{1}{p_i}-P_L(p_i)} \leq k E_{L}\pth{\phi,\qth{0,\frac{c_1\log k}{n}}}
            \lesssim \frac{k}{n \log k}.
            \label{eq:biasE1}
        \end{equation}

        Next we consider the conditional variance of $ \calE_1 $. In view of the fact that 
%        $ I_1\subseteq\sth{i:p_i\le \frac{c_1\log k}{n}} $ and 
        the standard deviation of sum of random variables is at most the sum of individual standard deviations, we obtain that
        % \nbr{PK: fix align!}
        \begin{align*}
            \var\qth{\calE_1|I_1}
            =& \sum_{i\in I_1}\var\qth{\phi(p_i)-g_L(N_i)}
            \le \sum_{i:p_i\le \frac{c_1\log k}{n}}\var\qth{g_L(N_i)}\\
            =& \sum_{i:p_i\le \frac{c_1\log k}{n}}\var\qth{\sum_{m\ne 1}\frac{a_m}{\pth{c_1\log k}^{m-1}}\frac{(N_i)_{m}}{n}
                +\pth{a_1+\log\frac{n}{c_1\log k}}\frac{N_i}{n}}\\
            \le&\frac{1}{n^2}\sum_{i:p_i\le \frac{c_1\log k}{n}}\pth{\sum_{m\ne 1}\frac{\abs{a_m}}{\pth{c_1\log k}^{m-1}}\sqrt {\var(N_i)_{m}}
                +\abs{a_1+\log\frac{n}{c_1\log k}}\sqrt {\var(N_i)}}^2.
        \end{align*}
        Since $ 0 \leq \phi(x) \leq e^{-1}$ on $ [0,1] $ then $\sup_{0\leq x\leq 1}|p_L(x)-\phi(x)| = E_L(\phi,[0,1]) \leq e^{-1}$.
        Therefore $\sup_{0\leq x\leq 1} |p_L(x)| \leq 2e^{-1}$.
        From the proof of \cite[Lemma 2, p. 1035]{CL11} we know that the polynomial coefficients can by upper bounded by $ |a_m|\le 2e^{-1} 2^{3L} $. 
%        In the regime of interest that $ n $ is smaller than exponential order of $ k $, 
Since $\log n \leq C \log k$, we have
        $ \abs{a_1+\log\frac{n}{c_1\log k}}\lesssim 2^{3L} $.
        Therefore all polynomial coefficients can be upper bounded by a constant factor of $ 2^{3L} $.
        We also need the following lemma to upper bound the variance of $ (N_i)_{m} $:
        \begin{lemma}
            Suppose $ X\sim\Poi(\lambda) $ and $ (x)_m=\frac{x!}{(x-m)!} $.
            Then $ \var(X)_m $ is increasing in $ \lambda $ and
            \[
                \var (X)_m
                =\lambda^mm!\sum_{k=0}^{m-1} \binom{m}{k} \frac{\lambda^k}{k!}
%            \le \lambda^mm!m\pth{ \frac{\pth{2e}^{2\sqrt{\lambda m}}}{\pi\sqrt{\lambda m}} \vee 1}.
                \le (\lambda m)^m \pth{ \frac{\pth{2e}^{2\sqrt{\lambda m}}}{\pi\sqrt{\lambda m}} \vee 1}.
            \]
            \label{lmm:varfm-upper}
        \end{lemma}
        Recall that $ L=c_0\log k $. Let $c_0 \leq c_1$.
        The monotonicity in \prettyref{lmm:varfm-upper} yields that $ \var(N_i)_{m}\le \var(\tilde{N})_{m} $ where $ \tilde{N}\sim\Poi(c_1\log k) $ whenever $ p_i\le \frac{c_1\log k}{n} $.
        Applying the upper bound in \prettyref{lmm:varfm-upper} and in view of the relation that $ m\le c_0\log k \le c_1\log k $, the conditional variance can be further upper bounded by the following 
        \begin{align}
            \var\qth{\calE_1|I_1}
            \lesssim & \frac{k}{n^2}\pth{\sum_{m=0}^{L}\frac{2^{3L}}{\pth{c_1\log k}^{m-1}}\sqrt {((c_1\log k)(c_1\log k))^m(2e)^{2\sqrt{(c_0\log k) (c_1\log k)}}}}^2\nonumber\\
            = & \frac{k}{n^2}\pth{\sum_{m=0}^{L}k^{(c_0\log 8+\sqrt{c_0c_1}\log(2e))}c_1\log k}^2\nonumber \\
            \lesssim & \frac{(\log k)^4}{n^2}k^{1+2(c_0\log 8+\sqrt{c_0c_1}\log(2e))}.\label{eq:varE1}
        \end{align}
        From \prettyref{eq:biasE1}--\prettyref{eq:varE1} we conclude that
        \begin{equation}
            \Expect[\calE_1^2] = \expect{\Expect[\calE_1|I_1]^2 + \var(\calE_1|I_1)} \lesssim \pth{\frac{k}{n\log k}}^2
            \label{eq:MSE_E1}
        \end{equation}
%        for sufficiently large $ k $ and sufficiently small $ c_0 $ such that 
as long as
        \begin{equation}
    c_0\log 8+\sqrt{c_0c_1}\log(2e)< \frac{1}{4}.
    \label{eq:c0}
\end{equation}

        % Let $ \calE_1(i)=\phi(p_i)-g_L(N_i) $, then $ \calE_1=\sum_{i\in I_1}\calE_1(i) $. Some basic algebra shows that
        % \[
        % \var(\calE_1)
        % =\sum_{i=1}^{k}\var\pth{\calE_1(i)}\Prob(i\in I_1)+(\Expect\calE_1(i))^2\Prob(i\in I_1)\Prob(i\notin I_1).
        % \]
        
        \paragraph{Case 2: Bias-corrected plug-in estimator}
        First note that $\calE_2$ can be written as
        \begin{equation}
            \calE_2=\sum_{i\in I_2}\pth{(p_i-\hat{p}_i)\log\frac{1}{p_i}+\hat{p}_i\log\frac{\hat{p}_i}{p_i}-\frac{1}{2n}},
            \label{eq:E2}
        \end{equation}
        where $ \hat{p}_i=\frac{N_i}{n} $ is an unbiased estimator of $ p_i $ since  $ N_i\sim\Poi(np_i) $. The first term is thus unbiased conditioned on $I_2$. Note the following elementary bounds on the function $x \log x$: 
        \begin{lemma}
    For any $x>0$,
        % We need the following bounds inspired from Taylor's expansion of $ x\log x $ at $ x=1 $:
        \begin{align*}
            0 \leq x\log x - (x-1) - \frac{1}{2}(x-1)^2 + \frac{1}{6}(x-1)^3 \leq \frac{(x-1)^4}{3}.
        \end{align*}
    \label{lmm:xlogx}
\end{lemma}
    Applying the above facts to $ x=\frac{\hat{p}_i}{p_i} $, we obtain that
    \begin{align*}
        \sum_{i\in I_2}p_i\frac{\hat{p}_i}{p_i}\log\frac{\hat{p}_i}{p_i}
        \ge&\sum_{i\in I_2}(\hat{p}_i-p_i)+\frac{(\hat{p}_i-p_i)^2}{2p_i}-\frac{(\hat{p}_i-p_i)^3}{6p_i^2}, \\ %\label{eq:pk1}\\
        \sum_{i\in I_2}p_i\frac{\hat{p}_i}{p_i}\log\frac{\hat{p}_i}{p_i}
        \le&\sum_{i\in I_2}(\hat{p}_i-p_i)+\frac{(\hat{p}_i-p_i)^2}{2p_i}
        -\frac{(\hat{p}_i-p_i)^3}{6p_i^2}+\frac{(\hat{p}_i-p_i)^4}{3p_i^3}. % \label{eq:pk2}
        % 0 \leq \sum_{i\in I_2}p_i\frac{\hat{p}_i}{p_i}\log\frac{\hat{p}_i}{p_i}- \sum_{i\in I_2}(\hat{p}_i-p_i)+\frac{(\hat{p}_i-p_i)^2}{2p_i}-\frac{(\hat{p}_i-p_i)^3}{6p_i^2} \le&\sum_{i\in I_2}  \frac{(\hat{p}_i-p_i)^4}{3 p_i^3}. % \label{eq:pk2}
    \end{align*}
    Plugging the inequalities above into \prettyref{eq:E2} and taking expectation on both sides conditioned on $ I_2 $, using the central moments of Poisson distribution that $ \Expect(X-\Expect[X])^2=\lambda, \Expect(X-\Expect[X])^3=\lambda, \Expect(X-\Expect[X])^4=\lambda(1+3\lambda) $ when $ X\sim \Poi(\lambda) $, we obtain that
        \[
        -\sum_{i\in I_2}\frac{1}{6n^2p_i}
        \le \Expect\qth{\calE_2|I_2}
        \le \sum_{i\in I_2}\frac{1+3np_i}{3n^3p_i^2}-\frac{1}{6n^2p_i}.
        \]
        By definition, $ I_2\subseteq \{i:p_i>\frac{c_3\log k}{n}\} $ and $|I_2| \leq k$.
%        $I_2 $ has cardinality at most $ \frac{n}{c_3 \log k}\wedge k $.
        Hence, almost surely,
        \begin{equation}
            \abs{\Expect\qth{\calE_2|I_2}}
            \lesssim \sum_{i\in I_2}\frac{1}{n^2p_i} + \sum_{i\in I_2}\frac{1}{n^3 p_i^2}
            \lesssim 
%            \frac{1}{(\log k)^2}\wedge\frac{k}{n\log k} \leq 
            \frac{k}{n\log k}.
            \label{eq:biasE2}
        \end{equation}
%        Note that $ p_i>\frac{c_3\log k}{n} $ and $ 
        
        It remains to bound the variance of the plug-in estimator. Note that
        \begin{align}
            \var\qth{\calE_2|I_2}
            \le \sum_{i:p_i>\frac{c_3\log k}{n}}\var\qth{\phi(p_i)-\phi(\hat{p}_i)}
            \le \sum_{i:p_i>\frac{c_3\log k}{n}}\Expect\pth{\phi(p_i)-\phi(\hat{p}_i)}^2.
            \label{eq:varE2-1}
        \end{align}
        % We will show in the following to conclude that $ \Expect\qth{\phi(p_i)-\phi(\hat{p}_i)}^2\le ... $ for any $ p_i>\frac{c_3\log k}{n} $:
%        
%        % 1. $ \hat{p}_i<p_i/2 $:
%        % \[
%        % \Expect[\phi(\hat{p_i})-\phi(p_i)]^2\indc{\hat{p}_i<p_i/2}
%        % \lesssim \prob{N_i<np_i/2}
%        % \le \frac{1}{k^{c_3\pth{1-\frac{1}{2}\log 2e}}},
%        % \]
%        % since $ \phi $ is bounded on $ [0,1] $ and $ np_i>c_3\log k $.
%
%        % 2. $ \hat{p}_i\in [p_i/2,2p_i] $:
%        % \[
%        % \Expect[\phi(\hat{p}_i)-\phi(p_i)]^2\indc{\hat{p}_i\in [p_i/2,2p_i]}
%        % \le \Expect[(\hat{p}_i-p_i)\sup_{x\in [p_i/2,2p_i] }\phi'(x)]^2
%        % \lesssim \pth{\log\frac{1}{p_i}}^2\frac{p_i}{n}.
%        % \]
%
%        % 3. $ \hat{p}_i>2p_i $:
%        % \[
%        % \Expect[\phi(\hat{p}_i)-\phi(p_i)]^2\indc{\hat{p}_i>2p_i}
%        % \le ....\nb{?}
%        % \]
%
%        % \nb{Just thought of another approach. TODO: verify if correct}  
%
        In view of the fact that $ \log x\le x-1 $ and $ x\log x\ge x-1 $ for any $ x>0 $, we have
        $$ \hat{p}_i-p_i=p_i\pth{\frac{\hat{p}_i}{p_i}-1}\le p_i\frac{\hat{p}_i}{p_i}\log\frac{\hat{p}_i}{p_i}= \hat{p}_i\log\frac{\hat{p}_i}{p_i}\le \hat{p}_i\pth{\frac{\hat{p}_i}{p_i}-1}=\hat{p}_i-p_i+\frac{(\hat{p}_i-p_i)^2}{p_i}  .$$
        Recall that $ \phi(p_i)-\phi(\hat{p}_i)=(p_i-\hat{p}_i)\log\frac{1}{p_i}+\hat{p}_i\log\frac{\hat{p}_i}{p_i} $.
        Then, by triangle inequality,
        % Within a constant factor we can obtain the following upper bound:
        \begin{align*}
            (\phi(p_i)-\phi(\hat{p}_i))^2
            \leq ~ & 2 (p_i-\hat{p}_i)^2 \log^2\frac{1}{p_i} +2 \pth{\hat{p}_i\log\frac{\hat{p}_i}{p_i}}^2\\
            \leq ~ & 2 (p_i-\hat{p}_i)^2 \log^2\frac{1}{p_i} + 4 (\hat{p}_i-p_i)^2+ \frac{4(\hat{p}_i-p_i)^4}{p_i^2}.
        \end{align*}
        Taking expectation on both sides yields that
        \[
        \Expect (\phi(p_i)-\phi(\hat{p}_i))^2
        \le \frac{2 p_i}{n}\pth{\log\frac{1}{p_i}}^2+\frac{4 p_i}{n}+\frac{12}{n^2} + \frac{4}{n^3 p_i}.
        \]
        Plugging the above into \prettyref{eq:varE2-1} and summing over $i$ such that $p_i \geq \frac{c_3\log k}{n}$, we have
        \begin{equation}
            \var[\calE_2|I_2]
            \lesssim \frac{(\log k)^2}{n}+ \frac{k}{n^2}
            \label{eq:varE2}
        \end{equation}
        where we used the fact that $\sup_{P \in \calM_k} \sum_{i=1}^k p_i \log^2 \frac{1}{p_i} \lesssim \log^2 k$.
        Assembling \prettyref{eq:biasE2}--\prettyref{eq:varE2} yields that
        \begin{equation}
            \Expect \calE_2^2
            % \lesssim \pth{\frac{1}{\log^4 k}\wedge \pth{\frac{k}{n\log k}}^2}+\frac{\log^2k}{n}.
            \lesssim \pth{\frac{k}{n\log k}}^2+\frac{\log^2k}{n}.
            \label{eq:MSE_E2}
        \end{equation}

        % Combining \prettyref{eq:MSE_E1} and \prettyref{eq:MSE_E2} in two cases into \prettyref{eq:MSE}, we obtain
        % \begin{align}
        %     \Expect(H(P)-\tilde{H})^2\Indc_E
        %     \le 2\Expect\calE_1^2
        %     +2\Expect\calE_2^2
        %     \lesssim \pth{\frac{k}{n\log k}}^2+\frac{(\log k)^2}{n}.
        %     \label{eq:MSE_E3}
        % \end{align}
        By assumption, $\log n \leq C \log k$ for some constant $C$.
        Choose $c_1>c_2>c_3>0$ such that $ c_1-c_2\log\frac{ec_1}{c_2}-1>C $ and $ c_3+c_2\log\frac{ec_2}{c_3}-1>C $ hold simultaneously, \eg, $ c_1=4(C+1), c_2=e^{-1}c_1 $, $ c_3=e^{-2}c_1 $, 
%        Plugging the MSE \prettyref{eq:MSE_E1} and \prettyref{eq:MSE_E2} in two cases and \prettyref{eq:E1c} -- \prettyref{eq:E2c} into \prettyref{eq:MSE}, 
and $c_0 \leq c_1$ satisfying the condition \prettyref{eq:c0}, \eg, $c_0=\frac{1}{300c_1}\wedge c_1\wedge 0.01$.
Plugging \prettyref{eq:MSE_E1}, \prettyref{eq:MSE_E2}, \prettyref{eq:E1c} and \prettyref{eq:E2c} into \prettyref{eq:MSE},
        we complete the proof.
    \end{proof}

    \begin{proof}[Proof of \prettyref{lmm:varfm-upper}]
        First we compute $ \Expect{(X)_m^2} $:
        \begin{align}
            \Expect{(X)_m^2}
            =&\sum_{x=0}^{\infty}\frac{e^{-\lambda}\lambda^x}{x!}\frac{x!^2}{(x-m)!^2}
            =\sum_{j=0}^{\infty}\frac{e^{-\lambda}\lambda^{j+m}}{j!}\frac{(j+m)!}{j!}
            =\lambda^mm!\Expect{\binom{X+m}{X}}\nonumber\\
            =&\lambda^mm!\expect{\sum_{k=0}^{m}\binom{m}{k}\binom{X}{X-k}}
            =\lambda^mm!\sum_{k=0}^{m}\binom{m}{k}\frac{\Expect{(X)_k}}{k!}
            =\lambda^mm!\sum_{k=0}^{m}\binom{m}{k}\frac{\lambda^k}{k!},
            \label{eq:Lag-m}
        \end{align}
%         since we know that a random variable  has factorial moment $ \Expect{(X)_k}=\lambda^k $. 
where we have used $ \Expect{(X)_k}=\lambda^k $. 
         Therefore the variance of $(X)_m$ is
        \[
        \var (X)_m
        =\lambda^mm!\sum_{k=0}^m \binom{m}{k} \frac{\lambda^k}{k!}-\lambda^{2m}
        =\lambda^mm!\sum_{k=0}^{m-1} \binom{m}{k} \frac{\lambda^k}{k!}
        \le \lambda^mm!\sum_{k=0}^{m-1}\frac{(\lambda m)^k}{(k!)^2}.
        \]
        The monotonicity of $\lambda \mapsto \var (X)_m$ follows from the equality part immediately. Since the maximal term in the summation is attained at $ k^*=\Floor{\sqrt{\lambda m}} $, we have
        \[
        \var (X)_m
        \le \lambda^mm!m\frac{(\lambda m)^{k^*}}{(k^*!)^2} \leq (\lambda m)^m \frac{(\lambda m)^{k^*}}{(k^*!)^2}
        \]

        If $ \lambda m<1 $ then $ k^*=0 $ and $ \frac{(\lambda m)^{k^*}}{(k^*!)^2}=1 $;
        otherwise $ \lambda m\ge 1 $ and hence $ \frac{\sqrt{\lambda m}}{2}<k^* \le \sqrt{\lambda m} $.
        Applying $ k^*!>\sqrt{2\pi k^*}\pth{\frac{k^*}{e}}^{k^*} $ yields 
        \[
        \frac{(\lambda m)^{k^*}}{(k^*!)^2}
        \le \frac{(\lambda m)^{k^*}}{2\pi\frac{\sqrt{\lambda m}}{2}\pth{\frac{\lambda m}{4e^2}}^{k^*}}
        =\frac{\pth{2e}^{2\sqrt{\lambda m}}}{\pi\sqrt{\lambda m}}. \qedhere
        \]
    \end{proof}
    
    \begin{remark}
    Note that the right-hand side of \prettyref{eq:Lag-m} coincides with $\lambda^mm!L_m(-\lambda)$, where $ L_m$ denotes the Laguerre polynomial of degree $ m $.
        The term $e^{\sqrt{\lambda m}}$ agrees with the sharp asymptotics of the Laguerre polynomial on the negative axis \cite[Theorem  8.22.3]{orthogonal.poly}.
%   \label{rmk:}
\end{remark}

\begin{proof}[Proof of \prettyref{lmm:xlogx}]
    It follows from Taylor's expansion of $ x\mapsto x\log x $ at $ x=1 $ that
    \[
    x\log x=(x-1)+\frac{1}{2}(x-1)^2-\frac{1}{6}(x-1)^3+\frac{1}{3}\int_1^x  \pth{\frac{x}{t}-1}^3\diff t. 
    \]
    Hence it suffices to show $ 0\le \int_1^x  \pth{\frac{x}{t}-1}^3\diff t \le (x-1)^4 $ for all $ x>0 $.
    If $ x>1 $, the conclusion is obvious since the integrand is always positive and no greater than $ (x-1)^3 $. 
    If $ x<1 $, we rewrite the integral as $ \int_x^1  \pth{1-\frac{x}{t}}^3\diff t $. Then the conclusion follows from the same reason that
    the integrand is always positive and at most $ (1-x)^3 $.
\end{proof}

    % \begin{proof}[Proof of \prettyref{lmm:GKL-bd}]
    %     The desired inequality obviously holds if $ q=0 $ for some $ i $. 
    %     The proof is based on the Taylor's expansion of $ x\log x $ at $ x=1 $.
    %     For one direction, we first observe that
    %     \[
    %     x\log x\ge (x-1)+\frac{1}{2}(x-1)^2-\frac{1}{6}(x-1)^3.
    %     \]
    %     It follows from above inequality that
    %     \[
    %     \sum_i q_i\frac{p_i}{q_i}\log \frac{p_i}{q_i}
    %     \ge \sum_i(p_i-q_i)+\frac{(p_i-q_i)^2}{2q_i}-\frac{(p_i-q_i)^3}{6q_i^2}.
    %     \]
        
    %     For the other direction, we observe that
    %     \[
    %     x\log x\le (x-1)+\frac{1}{2}(x-1)^2-\frac{1}{6}(x-1)^3+(x-1)^4.
    %     \]
    %     Similarly we obtain that
    %     \[
    %     \sum_i q_i\frac{p_i}{q_i}\log \frac{p_i}{q_i}
    %     \le \sum_i(p_i-q_i)+\frac{(p_i-q_i)^2}{2q_i}-\frac{(p_i-q_i)^3}{6q_i^2}+\frac{(p_i-q_i)^4}{q_i^3}.
    %     \]
    % \end{proof}
    
%   \begin{lemma}
%     \nb{PY please finish this one}
%     \[
%     L_n(-\lambda) \leq \exp(2 \sqrt{\lambda n}) poly(n) term...
%     \]
%     \label{lmm:lag}
% \end{lemma}
% \begin{proof}
% $L_n(-\lambda) = \sum_{k=0}^n \binom{n}{k} \frac{\lambda^k}{k!} \leq \sum_{k=0}^n a_k$, where 
% $a_k \triangleq \frac{(\lambda n)^k}{(k!)^2}$. 
%  Since $\frac{a_k}{a_{k-1}} = \frac{\lambda n}{k^2}$. Therefore the maximal $a_k$ is attained at $k=\lfloor \sqrt{\lambda n} \rfloor$. Hence
%  \[
%  L_n(-\lambda) \leq (n+1) \frac{(\lambda n)^{\lfloor \sqrt{\lambda n} \rfloor}}{(\lfloor \sqrt{\lambda n} \rfloor!)^2} \leq \nb{use non-asym Stirling approximation}
%  \]
% \end{proof}

\section{Non-asymptotic risk bounds for the plug-in estimator}
\label{app:plug}
\begin{proof}[Proof of \prettyref{prop:Rplug-rate}]
Recall the worst-case quadratic risk of the plug-in estimator $\Rplug(k,n)$ defined in \prettyref{eq:Rplug}.
We show that for any $k \geq 2$ and $n\geq 2$,
\begin{equation}
    \pth{\frac{k}{n} \wedge 1}^2 + \frac{\log^2 k}{n} \lesssim   \Rplug(k,n) \lesssim \pth{\frac{k}{n}}^2  + \frac{\log^2 (k\wedge n)}{n}
    \label{eq:Rplug-bd}
\end{equation}

The second term of the lower bound follows from the minimax lower bound \prettyref{prop:lb1} which applies to all $k$ and $n$. 
To prove the first term of lower bound, we take $ P $ as uniform distribution. We consider its bias here since squared bias is a lower bound for MSE.
We denote the empirical distribution as $ \hat{P}=\frac{N}{n} $. Applying Pinsker's inequality and Cauchy-Schwarz inequality, we obtain
\begin{align*}
    \Expect(\Hplug(N) - H)
    =&-\Expect[D(\hat{P}||P)]
    \le -2\Expect[(\TV(\hat{P},P))^2]\\
    \le& -2(\Expect[\TV(\hat{P},P)])^2
    =-2\pth{\frac{k}{2n}\Expect\abs{N_1-\frac{n}{k}}}^2,
\end{align*}
where $ N_1\sim \Binom\pth{n,\frac{1}{k}} $. 
From \cite[Theorem 1]{berend2013sharp}, we know that
$ \Expect\abs{N_1-\frac{n}{k}}=\frac{2n}{k}\pth{1-\frac{1}{k}}^n $ when $ n<k $ 
and $ \Expect\abs{N_1-\frac{n}{k}}\ge \sqrt{\frac{n}{2k}\pth{1-\frac{1}{k}}} $ when $ n\ge k $.
Therefore
\begin{align*}
    &-\Expect(\Hplug(N) - H)
    \ge 2\pth{1-\frac{1}{k}}^{2n}
    \gtrsim 1,\quad n<k,
%    \label{eq:p-con-1}
\\
    &-\Expect(\Hplug(N) - H)
    \ge \frac{k}{4n}\pth{1-\frac{1}{k}}
    \gtrsim \frac{k}{n},\quad n\ge k.
%    \label{eq:p-con-2}
\end{align*}
Consequently, 
%From \prettyref{eq:p-con-1}--\prettyref{eq:p-con-2} we conclude that
\[
\Expect[(\Hplug(N) - H)^2]
\ge [\Expect(\Hplug(N) - H)]^2
\gtrsim \pth{\frac{k}{n} \wedge 1}^2
\]

The upper bound of MSE follows from the upper bounds of bias and variance.
The squared bias can be upper bounded by $ (\frac{k-1}{n})^2 $ according to \cite[Proposition 1]{Paninski03}. 
For the variance we apply Steele's inequality \cite{steele86}:
\begin{equation}
    \var[\Hplug]\le \frac{n}{2}\Expect(\Hplug(N)-\Hplug(N'))^2,
    \label{eq:steel-H}
\end{equation}
where $ N' $ is the histogram of $(X_1,\ldots,X_{n-1},X_n')$ and $X_n'$ is an independent copy of $X_n$.
%samples given by replacing $ X_n $ with an \iid copy $ X_n' $.
Let $ \tilde N=(\tilde N_1,\dots,\tilde N_k) $ be the histogram of $ X_1^{n-1} $, then $ \tilde N\sim \Multinom(n-1,P) $ independently of $ X_n, X_n' $.
Hence, applying triangle inequality,
\begin{align}
    &\Expect(\Hplug(N)-\Hplug(N'))^2\nonumber\\
    =&\expect{\expect{\pth{\phi\pth{\frac{\tilde N_{X_n}+1}{n}}-\phi\pth{\frac{\tilde N_{X_n}}{n}}
                +\phi\pth{\frac{\tilde N_{X_n'}}{n}}-\phi\pth{\frac{\tilde N_{X_n'}+1}{n}}}^2\big|X_n,X_n'}}\nonumber\\
    \le & 4 \sum_{j=1}^{k}\expect{\pth{\phi\pth{\frac{\tilde N_{j}+1}{n}}-\phi\pth{\frac{\tilde N_{j}}{n}}}^2}p_j
    = \frac{4}{n^2}\sum_{j=1}^{k}\expect{\pth{\tilde N_{j}\log(1+\tilde N_{j}^{-1})+\log\frac{\tilde N_{j}+1}{n}}^2}p_j\nonumber\\
    \le & \frac{8}{n^2}+\frac{8}{n^2}\sum_{j=1}^{k}\expect{\log^2\frac{\tilde N_{j}+1}{n}}p_j,\label{eq:Hplug-diff}
\end{align}
where the last step follows from $ 0\le x\log(1+x^{-1})\le 1 $ for all $ x>0 $.

Now we rewrite and upper bound the last expectation:
\begin{align}
    \expect{\log^2\frac{\tilde N_{j}+1}{n}}
    =& \expect{\log^2\frac{n}{\tilde N_{j}+1}\indc{\tilde N_j\le (n-1)p_j/2}} + \expect{\log^2\frac{n}{\tilde N_{j}+1}\indc{\tilde  N_j> (n-1)p_j/2}}\nonumber\\
    \le& (\log^2n)\,\prob{\tilde N_j\le (n-1)p/2}+\log^2\frac{2n}{(n-1)p_j}.\label{eq:log-decomp}
\end{align}
Applying Chernoff bound for Binomial tail \cite[Theorem 4.5]{MU06} and plugging into \prettyref{eq:Hplug-diff} then \prettyref{eq:steel-H}, we obtain
\begin{align*}
    \var\Hplug
    \lesssim&~ \frac{1}{n}+\frac{1}{n}\sum_{j=1}^{k}p_j(\log^2p_j+ \log^2n \exp(-(n-1)p_j/8))\\
    \lesssim&~ \frac{\log^2k}{n}+\frac{\log^2n}{n}\frac{k}{n} = \frac{\log^2k}{n}\pth{1+\frac{k\log^2n}{n\log^2k}}
\end{align*}
where we have used $ \sum_{i=1}^{k} p_i\log^2 p_i\lesssim \log^2k $ and $\sup_{x>0} x \exp(-(n-1)x/8) = \frac{8}{(n-1)e} $.
We know that $ \frac{k\log^2n}{n\log^2k}\lesssim 1 $ when $ n\ge k $ and thus $ \var\Hplug\lesssim \frac{\log^2k}{n} $.
From \cite[Remark (iv), p. 168]{AK01} we also know that $ \var\Hplug(N)\lesssim \frac{\log^2 n}{n} $ for all $ n $ and consequently $ \var\Hplug(N)\lesssim \frac{\log^2(k\wedge n)}{n} $.
\end{proof}

\section{Moment matching and best polynomial approximation}
\label{app:moments}

In this appendix we discuss the relationship between moment matching and best polynomial approximation and, in particular, provide a short proof of \prettyref{eq:RE}. Let $g$ be a continuous function on the interval $[a,b]$.
Abbreviate by $\hat{\calE}^*$ the best uniform approximation error 
$E_L(g, [a,b]) =\inf_{p\in\calP_L}\sup_{x\in [a,b]}\abs{ g(x)-p(x) }$.

Let $ \calS_L=\sth{ (X,X')\in [a,b]^2: \expect{X^j} = \expect{X'^j},j = 1,\ldots,L } $. For any polynomial $ p\in \calP_L $, we have
\begin{align*}
    \calE^* & \triangleq \sup_{(X,X')\in \calS_L}{ \expect{ g(X) } - \expect{ g(X') } } \\
    & = \sup_{(X,X')\in \calS_L}{ \expect{ g(X)-p(X)} - \expect{ g(X')-p(X')} },
\end{align*}
and therefore by triangle inequality
\begin{align*}
    \calE^* & = \inf_{p\in\calP_L}\sup_{(X,X')\in \calS_L}{ \expect{ g(X)-p(X)} - \expect{ g(X')-p(X')} }\\
    & \le 2 \inf_{p\in\calP_L} \sup_{x\in [a,b]}\abs{ g(x)-p(x) } = 2 E_L(g, [a,b]).
\end{align*}

For the achievability part, Chebyshev alternating theorem \cite[Theorem 1.6]{petrushev2011rational} states that 
there exists a (unique) polynomial $ p^*\in\calP_L $ and at least $ L+2 $ points $ a\le x_1<\dots<x_{L+2}\le b $ and $ \alpha\in \sth{0,1} $
such that $ g(x_i)-p^*(x_i)=(-1)^{i+\alpha}\hat{\calE}^* $.
Fix any $ l=0,1,\dots,L $, define a Lagrange interpolation polynomial 
$ f_l(x)\triangleq \sum_{j=1}^{L+2}x_j^l\frac{\prod_{v\ne j}(x-x_v)}{\prod_{v\ne j}(x_j-x_v)} $ satisfying that $ f_l(x_j)=x_j^l $ for $ j=1,\dots,L+2 $.
Since $f_l$ has degree $L+1$, it must be that $ f_l(x)=x^l $.
Note that the coefficient of $ x^{L+1} $ of polynomial $ f_l $ is $ 0 $, \ie, $ \sum_ix_i^lb_i=0 $ where $ b_i\triangleq(\prod_{v\ne i}(x_i-x_v))^{-1} $.
Define $ w_i=\frac{2b_i}{\sum_j|b_j|} $, then $ \sum_i|w_i|=2 $.
When $ l=0 $ then $ \sum_ib_i=0 $ so $ \sum_iw_i=0 $.
Note that $ w_i $ change signs alternatively.
Construct discrete random variables $ X,X' $ with distributions $ \prob{X=x_i}=\abs{w_i} $ for $ i $ odd and $ \prob{X'=x_i}=\abs{w_i} $ for $ i $ even. Then $ (X,X')\in S_L $.
The property of those $ L+2 $ points that $ g(x_i)-p^*(x_i)=(-1)^{i+\alpha}\hat{\calE}^* $ yields that $ \abs{ \expect{ g(X)-p^*(X)} - \expect{ g(X')-p^*(X')} }=2\hat{\calE}^* $.

\begin{remark}
    Alternatively, the achievability part can be argued from an optimization perspective (zero duality gap, see \cite[Exercise 8.8.7, p. 236]{Luenberger69}), or using the Riesz representation of linear operators as in \cite{DL93}, which has been used in \cite{LNS99} and \cite{CL11}.
\end{remark}

\section{Best polynomial approximation of the logarithm function}
\label{app:error}

\begin{proof}[Proof of \prettyref{lmm:sep}]
    Recall the best uniform polynomial approximation error $ E_m(f, I)$ defined in \prettyref{eq:EL}. Put $ E_m(f) \triangleq E_m(f,[-1,1]) $. 
    In the sequel we shall slightly abuse the notation by assuming that $c L\in\naturals$, for otherwise the desired statement holds with $c$ replaced by $c/2$.
    % , in view of the fact that $m \mapsto E_m(f,I) $ is decreasing.
    % For all $ c\le 10^{-4} $ and $ L>400 \log \frac{1}{2\pi c} $, $ E_{cL}(\log, [1/L^2,1])\ge \frac{1}{200} $.
    % Equivalently, for $ \eta_L = (c/L)^2 $, $ E_{L}(\log, [\eta_L,1])=\Omega(1) $.
    Through simple linear transformation we see that $ E_{cL}(\log,[L^{-2},1])=E_{cL}(f_L) $
    where $$ f_L(x)=-\log \pth {\frac{1+x}{2}+\frac{1-x}{2L^2}} .$$
%In order to show that the desired result $E_{cL}(f_L) \gtrsim 1$ makes sense, 

The difficulty in proving the desired 
\begin{equation}
E_{cL}(f_L) \gtrsim 1   
        \label{eq:sep1}
\end{equation}
lies in the fact that the approximand $f_L$ changes with the degree $L$. In fact, the following asymptotic result has been shown in \cite[Section 7.5.3, p.~445]{timan63}: $E_L(\log(a-x)) = \frac{1+o(1)}{L\sqrt{a^2-1} (a+\sqrt{a^2-1})^L}$ for \emph{fixed} $a>1$ and $L\diverge$. In our case $E_{cL}(f_L) = E_{cL}(\log(a-x))$ with $a=\frac{1+L^{-2}}{1-L^{-2}}$. The desired \prettyref{eq:sep1} would follow if one substituted this $a$ into the asymptotic expansion of the approximation error, which, of course, is not a rigorous approach. To prove \prettyref{eq:sep1}, we need non-asymptotic lower and upper bounds on the approximation error.     
There exist many characterizations of approximation error, such as Jackson's theorem, in term of various moduli of continuity of the approximand. 
    Let $\Delta_m(x)=\frac{1}{m}\sqrt{1-x^2}+\frac{1}{m^2}$ and define the following modulus of continuity for $f$  (see, \eg, \cite[Section 3.4]{petrushev2011rational}):
    $$ \tau_1(f,\Delta_m)=\sup\{|f(x)-f(y)|: x,y\in[-1,1],|x-y|\le \Delta_m(x)\} .$$
%    Initiated from Jackson's theorem, there exist many characterizations of approximation error in term of refined modulus of continuity. We choose $ \tau_1 $ from \cite[Section 3.4]{petrushev2011rational} for our proof. 
    We first state the following bounds on $ \tau_1 $ for $f_L$:

    \begin{lemma}[Direct bound]
        \begin{equation}
            \tau_1(f_L,\Delta_m) \le \log \pth{\frac{2L^2}{m^2}},~\forall m\le 0.1L.
            \label{eq:direct}
        \end{equation}
        \label{lmm:direct}
    \end{lemma}

    \begin{lemma}[Converse bound]
        \begin{equation}
            \tau_1(f_L,\Delta_L)\ge 1, \forall L\ge 10.
            \label{eq:converse}
        \end{equation}
        \label{lmm:converse}
    \end{lemma}

    From \cite[Theorem 3.13, Lemma 3.1]{petrushev2011rational} we know that $ E_m(f_L)\le 100\tau_1(f_L,\Delta_m) $. Therefore, for all $ c\le 10^{-7}<0.1 $, the direct bound in \prettyref{lmm:direct} gives us
    \begin{equation}
        \frac{1}{L}\sum_{m=1}^{cL}E_m(f_L)
        \le \frac{100}{L} \sum_{m=1}^{cL}\log \pth{\frac{2L^2}{m^2}}
        = 100c\log 2+ \frac{200}{L}\log \frac{L^{cL}}{(cL)!}
        < \frac{1}{400}-\frac{100}{L}\log(2\pi cL),
        \label{eq:directsum}
    \end{equation}
    where the last inequality follows from Stirling's approximation $ n!>\sqrt{2\pi n}(n/e)^n $.
    We apply the converse result for approximation in \cite[Theorem 3.14]{petrushev2011rational} that
    \begin{equation}
        \tau_1(f_L,\Delta_L)\le \frac{100}{L}\sum_{m=0}^{L}E_m(f_L) ,
        \label{eq:convsum}
\end{equation}
     where $ E_0(f_L)=\log L $.     
    %    we apply \prettyref{eq:directsum} and \prettyref{lmm:converse} to obtain 
Assembling \prettyref{eq:converse}--\prettyref{eq:convsum}, we obtain for all $ c\le 10^{-7} $ and $ L>10 \vee \pth{100\times 400 \log \frac{1}{2\pi c}} $, 
    \[
    \frac{1}{L}\sum_{m=cL+1}^{L}E_m(f_L)
    \ge \frac{1}{100} - \pth {\frac{1}{L}E_0(f_L) + \frac{1}{L}\sum_{m=1}^{cL}E_m(f_L)}
    \ge \frac{1}{100} - \pth {\frac{1}{400} + \frac{100\log \frac{1}{2\pi c}}{L}}
    > \frac{1}{200}.
    \]
    By definition, the approximation error $ E_m(f_L) $ is a decreasing function of the degree $ m $.
    Therefore for all $ c\le 10^{-7} $ and $ L>4\times 10^4 \log \frac{1}{2\pi c} $,
    \[
    E_{cL}(f_L)
    \ge \frac{1}{L-cL}\sum_{m=cL+1}^{L}E_m(f_L)
    \ge \frac{1}{L}\sum_{m=cL+1}^{L}E_m(f_L)
    \ge \frac{1}{200}. \qedhere
    \]

\end{proof}

% \begin{remark}
%     The above proof intends to show $ E_{cL}(\log,[1/L^2,1])\gtrsim 1 $ for some $ c $ when $ L $ is large,
%     the constants in the proof are very loose to obtain concise non-asymptotic bounds.
%     Actually from numerical analysis, $ E_{L}(\log,[1/L^2,1])\approx\nbr{Experiment} $ when $ L $ is large enough.
% \end{remark}
\begin{remark}
    From the direct bound \prettyref{lmm:direct} we know that $ E_{cL}(\log,[1/L^2,1])\lesssim 1  $. Therefore the bound \prettyref{eq:sep} is in fact tight: $ E_{cL}(\log,[1/L^2,1])\asymp 1  $.
\end{remark}

\begin{proof}[Proof of Lemmas \ref{lmm:direct} and \ref{lmm:converse}]

    First we show \prettyref{eq:direct}.
    Note that
    \begin{equation*}
        \tau_1(f_L,\Delta_m)
        =  \sup_{x \in \qth{-1,1}}\sup_{y:|x-y|\le\Delta_m(x)}|f_L(x)-f_L(y)|.
    \end{equation*}
    For fixed $ x\in[-1,1] $, to decide the optimal choice of $ y $ we need to consider whether $ \xi_1(x)\triangleq x-\Delta_m(x)\ge -1 $ and whether $ \xi_2(x)\triangleq x+\Delta_m(x)\le 1 $.
    Since $ \xi_1 $ is convex, $ \xi_1(-1)<-1 $ and $ \xi_1(1)>-1 $, then $ \xi_1(x)>-1 $ if and only if $ x>x_m $,
    % Observe the following equivalence relations:
    % \begin{equation}
    %     \begin{aligned}
    %         % &\sth{x\in [-1,1]:x-\Delta_m(x)<-1} \equiv \sth{x \in [-1,x_m)},\\
    %         % &\sth{x\in [-1,1]:x-\Delta_m(x)>-1} \equiv \sth{x \in (x_m,1]}, 
    %         &\sth{x\in [-1,1]:x-\Delta_m(x)<-1} = [-1,x_m),\\
    %         &\sth{x\in [-1,1]:x-\Delta_m(x)>-1} = (x_m,1], 
    %     \end{aligned}
    %     \label{eq:equiv-tau}
    % \end{equation}
    where $ x_m $ is the unique solution to $ \xi_1(x)=-1 $, given by
    \begin{equation}
        x_m=\frac{m^2-m^4+\sqrt{-m^2+3m^4}}{m^2+m^4}.
        \label{eq:xm}
    \end{equation}
    Note that $ \Delta_m $ is an even function and thus $ \xi_2(x)=-\xi_1(-x) $.
    Then $ \xi_2(x)<1 $ if and only if $ x<-x_m $.

    Since $ f_L $ is strictly decreasing and convex, for fixed $ x $ and $ d>0 $ we have $f_L(x-d)-f_L(x) > f_L(x)-f_L(x+d) > 0$ as long as $-1<x-d<x+d<1$.
    If $ m\ge 2 $ since $ \xi_1(0)>-1 $ then $ x_m<0 $ and $ -x_m>0 $.
    Therefore
    % can decompose the supremum and it turns out 
    \begin{equation*}
        \tau_1(f_L,\Delta_m) 
        % =  \sup_{x \in \qth{-1,1}}\sup_{y:|x-y|\le\Delta_m(x)}|f_L(x)-f_L(y)|\\
        =  \sup_{x<x_m}\sth{f_L(x)-f_L(\xi_2(x))}
        \vee \sup_{x<x_m}\sth{f_L(-1)-f_L(x)}
        \vee \sup_{x\ge x_m}\sth{f_L(\xi_1(x))-f_L(x)}.
    \end{equation*}
    Note that the second term in the last inequality is dominated by the third term since $ f_L(\xi_1(x_m))-f_L(x_m)=f_L(-1)-f_L(x_m) >  f_L(-1)-f_L(x) $ for any $ x<x_m $. Hence
    \begin{align}
        \tau_1(f_L,\Delta_m)= & \sup_{x\in [-1,x_m)}\sth{f_L(x)-f_L(\xi_2(x))} \vee \sup_{x\in [x_m,1]}\sth{f_L(\xi_1(x))-f_L(x)}\nonumber\\
        = & \sup_{x\in [-1,x_m)}\sth{\log \pth{1+\beta_L(x)}} \vee \sup_{x\in [x_m,1]}\sth{-\log \pth{1-\beta_L(x)}},\label{eq:tau}
    \end{align}
    where $ \beta_L(x)\triangleq \frac{\Delta_m(x)}{x+\frac{L^2+1}{L^2-1}} $.
    If $ m=1 $ we know that $ x_1>0 $ and $ -x_1<0 $ by \prettyref{eq:xm}, then
    \begin{equation*}
        \tau_1(f_L,\Delta_m) 
        % =  \sup_{x \in \qth{-1,1}}\sup_{y:|x-y|\le\Delta_m(x)}|f_L(x)-f_L(y)|\\
        =  \sup_{x<x_m}\sth{f_L(x)-f_L(\xi_2(x)\wedge 1)}
        \vee \sup_{x<x_m}\sth{f_L(-1)-f_L(x)}
        \vee \sup_{x\ge x_m}\sth{f_L(\xi_1(x))-f_L(x)}.
    \end{equation*}
    Since $ f_L(\xi_2(x)\wedge 1)\ge f_L(\xi_2(x)) $, by the same argument, \prettyref{eq:tau} remains a valid upper bound of $ \tau_1(f_L,\Delta_1) $.
    Next we will show separately that the two terms in \prettyref{eq:tau} both satisfy the desired upper bound. 

    For the first term in \prettyref{eq:tau}, note that
    \[
    \beta_L(x)
    = \frac{\frac{1}{m}\sqrt{1-x^2}+\frac{1}{m^2}}{x+1+\frac{2}{L^2-1}}
    \le \frac{1}{m^2}\frac{L\sqrt{1-x^2}+1}{\pth{x+1}+\frac{2}{L^2}}
    = \frac{L^2}{m^2}\frac{\sqrt{1-x^2}+\frac{1}{L}}{L\pth{x+1}+\frac{2}{L}}.
    \]
    One can verify that $ \frac{\sqrt{1-x^2}+\frac{1}{L}}{L\pth{x+1}+\frac{2}{L}}\le 1 $ for any $ x\in[-1,1] $.
    Therefore
    \begin{equation*}
        \log \pth{1+\beta_L(x)}\le \log\pth{1+\frac{L^2}{m^2}},~\forall x\in [-1,1]
    \end{equation*}
    and, consequently,
    \begin{equation}
        \sup_{x\in [-1,x_m)}\sth{\log \pth{1+\beta_L(x)}}
        \le \log\pth{\frac{2L^2}{m^2}},~\forall m\le L.
        \label{eq:tau-1}
    \end{equation}

    For the second term in \prettyref{eq:tau}, it follows from the derivative of $ \beta_L(x) $ that it is decreasing when $ x>\frac{1-L^2}{1+L^2} $.
    From \prettyref{eq:xm} we have $ x_m>\frac{1-m^2}{1+m^2} $ and hence $ x_m>\frac{1-L^2}{1+L^2} $ when $ m\le L $.
    So the supremum is achieved exactly at the left end of $ [x_m,1] $, that is:
    \begin{equation*}
        \sup_{x\in [x_m,1]}\sth{-\log \pth{1-\beta_L(x)}}
        =-\log \pth{1-\beta_L(x_m)}
        =\log\pth{\frac{1+x_m}{2}L^2+\frac{1-x_m}{2}}.
    \end{equation*}
    From \prettyref{eq:xm} we know that $ x_m\ge -1 $ and $  x_m<-1+\frac{3.8}{m^2} $. Therefore $ \frac{1-x_m}{2}\le 1 $ and $ \frac{x_m+1}{2}<\frac{1.9}{m^2} $.
    For $ m\le 0.1 L $, we have
    \begin{equation}
        \sup_{x\in [x_m,1]}\sth{-\log \pth{1-\beta_L(x)}}
        \le \log\pth{1+\frac{1.9m^2}{L^2}}
        \le \log\pth{\frac{2m^2}{L^2}}.
        \label{eq:tau-2}
    \end{equation}
    Plugging \prettyref{eq:tau-1} and \prettyref{eq:tau-2} into \prettyref{eq:tau}, we complete the proof of \prettyref{lmm:direct}.
%\end{proof}
%
%\begin{proof}[Proof of \prettyref{lmm:converse}]
%    We still use the notations in the proof of \prettyref{lmm:direct} since it is closely related.

Next we prove \prettyref{eq:converse}. Recall that $ x_L-\Delta_L(x_L)=-1 $.
    By definition,
    \[
    \tau_1(f_L,\Delta_L)
    \ge f_L(x_L-\Delta_L(x_L))-f_L(x_L)
    =\log\pth{\frac{1+x_L}{2}L^2+\frac{1-x_L}{2}}.
    \]
    Using the close-form expression of $ x_L $ in \prettyref{eq:xm} with $ m=L $, we further obtain
    \[
    \tau_1(f_L,\Delta_L)
    \ge \log\pth{\frac{2L^2+\sqrt{-L^2+3L^4}}{2(L^2+1)}+\frac{2L^4-\sqrt{-L^2+3L^4}}{2(L^2+L^4)}}
    \ge 1
    \]
    when $ L\ge 10 $.
\end{proof}

\section{Approximation error at the end points}
        \label{app:end}

    % By the definition $ P_L(0)=a_0\frac{c_1\log k}{n} $ and $ E_L(\phi,[0,\frac{c_1\log k}{n}])=\frac{c_1\log k}{n}E_L(\phi,[0,1]) $.
    % It suffices to show that $ a_0=E_L(\phi,[0,1]) $.
    % Recall that $ p_L(x) $ is the best polynomial to approximate $ \phi $ on $ [0,1] $.
    We prove the claim in \prettyref{rmk:adaptive}.
    By Chebyshev alternating theorem \cite[Theorem 1.6]{petrushev2011rational}, the error function $ g(x)\triangleq P_L(x)-\phi(x) $ attains uniform approximation error (namely, $\pm E_L(\phi)$) on at least $ L+2 $ points with alternative change of signs; moreover, these points must be stationary points or endpoints.
    Taking derivatives, $ g'(x)=P_L'(x)+\log(ex) $ and $ g''(x)=\frac{xP_L''(x)+1}{x} $.
Since    $ g'' $ has at most $ L-1 $ roots in $(0,1)$ and hence $ g' $ has at most $ L-1 $ stationary points,
the number of roots of $g'$ and hence the number of stationary points of $g$ in $(0,1)$ are at most $L$.
Therefore the error at the ends points must be maximal, \ie, $ |g(0)|=|g(1)|=E_L(\phi)$.
To determine the sign, note that $ g'(0)=-\infty $ then $g(0)$ must be positive for otherwise the value of $ g $ at the first stationary point is below $-E_L(\phi)$ which is a contradiction.
Hence $ a_0=g(0)=E_L(\phi)$.

\section*{Acknowledgment}
It is a pleasure to thank Yury Polyanskiy for pointing out \cite{VV10} and inspiring discussions in the early stage of the project, Tony Cai and Mark Low for many conversations on their result \cite{CL11}, R. Srikant for the pointers in \cite{Luenberger69}, and Jiantao Jiao for sharing his observation that the variance bound for the empirical entropy can be improved from  $\frac{\log^2 n}{n}$ to  $\frac{\log^2(k\wedge n)}{n}$. 
The authors are also grateful to the Associate Editor and anonymous reviewers for helpful comments.

%\bibliographystyle{alpha}
%\bibliography{IEEEabrv,strings,refs,entropy_ref}
\input{entropy.bbl}

\end{document}

%% file: entropy.bbl
\newcommand{\etalchar}[1]{$^{#1}$}